\documentclass{article}
\usepackage{authblk}
\usepackage{amsthm}
\usepackage{thmtools, thm-restate}
\usepackage[a4paper, margin=1in]{geometry} 
\usepackage{pifont}
\usepackage{pict2e}
\usepackage{amsmath}
\usepackage[title]{appendix}
\usepackage{tikz}
\usetikzlibrary {arrows.meta}
\usepackage[T1]{fontenc} %
\usepackage{lmodern} %
\rmfamily %
\title{Stable Hypergraph Matching in Unimodular Hypergraphs}

\author[1]{P{\'e}ter Bir{\'o}}
\author[1]{Gergely Cs{\'a}ji}
\author[1]{Ildik{\'o} Schlotter}
\affil[1]{
HUN-REN Centre for Economic and Regional Studies, Budapest, Hungary 
\{biro.peter, csaji.gergely, schlotter.ildiko\}@krtk.hun-ren.hu
}

\DeclareFontShape{T1}{lmr}{b}{sc}{<->ssub*cmr/bx/sc}{}
\DeclareFontShape{T1}{lmr}{bx}{sc}{<->ssub*cmr/bx/sc}{}
\usepackage{amssymb}
\usepackage{booktabs}
\usepackage[mathscr]{eucal}
\usepackage{indentfirst}
\usepackage{graphicx}
\usepackage{comment}
\usepackage{hyperref}
\usepackage{algorithm}
\usepackage[noend]{algpseudocode}

\usepackage{old-arrows}
\usepackage{stmaryrd}
\usepackage{enumitem}
\setlist{noitemsep}
\usepackage{rotating}
\usepackage{xcolor}
\colorlet{green}{green!50!black}
\hypersetup{
	colorlinks=true,
	citecolor=green,
	linkcolor=blue,  
	urlcolor=black,
}

\newcommand{\mycupdot}{\mathbin{\mathaccent\cdot\cup}}

\makeatletter
\providecommand*{\cupdot}{%
  \mathbin{%
    \mathpalette\@cupdot{}%
  }%
}
\newcommand*{\@cupdot}[2]{%
  \ooalign{%
    $\m@th#1\cup$\cr
    \sbox0{$#1\cup$}%
    \dimen@=\ht0 %
    \sbox0{$\m@th#1\cdot$}%
    \advance\dimen@ by -\ht0 %
    \dimen@=.5\dimen@
    \hidewidth\raise\dimen@\box0\hidewidth
  }%
}

\providecommand*{\bigcupdot}{%
  \mathop{%
    \vphantom{\bigcup}%
    \mathpalette\@bigcupdot{}%
  }%
}
\newcommand*{\@bigcupdot}[2]{%
  \ooalign{%
    $\m@th#1\bigcup$\cr
    \sbox0{$#1\bigcup$}%
    \dimen@=\ht0 %
    \advance\dimen@ by -\dp0 %
    \sbox0{\scalebox{2}{$\m@th#1\cdot$}}%
    \advance\dimen@ by -\ht0 %
    \dimen@=.5\dimen@
    \hidewidth\raise\dimen@\box0\hidewidth
  }%
}
\makeatother

\newcommand{\linkproof}[1]{%
    \hyperref[#1]{$\star$}%
}

\newcommand{\stabset}{T}
\newcommand{\remsetL}{X}
\newcommand{\remsetFIN}{Y}

\newcommand{\spr}{I'}

\allowdisplaybreaks
\newcommand{\decprob}[3]{
   \begin{center}%
   \fbox{
    \begin{minipage}{0.96\linewidth}%
      \textsc{#1}\\[4pt]
      \begin{tabular}{@{}l@{\hspace{4pt}}p{0.83\linewidth}}
      \textbf{Input:} & #2\\[0.2ex]
      \textbf{Question:} & #3
      \end{tabular}
    \end{minipage}%
    }
  \end{center}
  }
\newcommand{\searchprob}[3]{
   \begin{center}%
   \fbox{
    \begin{minipage}{0.96\linewidth}%
      \textsc{#1}\\[4pt]
      \begin{tabular}{@{}l@{\hspace{4pt}}p{0.83\linewidth}}
      \textbf{Input:} & #2\\[0.2ex]
      \textbf{Output:} & #3
      \end{tabular}
    \end{minipage}%
    }
  \end{center}
  }  
\theoremstyle{plain}
\newtheorem{theorem}{Theorem}[section]
\newtheorem{Lemma}[theorem]{Lemma}
\newtheorem{Cor}[theorem]{Corollary}
\newtheorem{proposition}[theorem]{Proposition}
\newtheorem{claim}[theorem]{Claim}

\theoremstyle{definition}
\newtheorem{Def}[theorem]{Definition}

\newtheorem{Rem}[theorem]{Remark}
\newtheorem{?}[theorem]{Problem}
\newtheorem{Ex}[theorem]{Example}

\newcommand{\M}{\mu}
\newcommand{\rcp}{subpath}
\newcommand{\capacuda}{c}
\newcommand{\capacquota}{q}

\def\tnph{\textbf{NPh}}
\def\txp{\textbf{XP}}
\def\twh{\textbf{W[1]h}}

\newcommand{\poly}{\mathrm{poly}}
\newcommand{\probfont}[1]{{\textmd{\textup{\textsc{#1}}}}}
\def\shm{\probfont{SHM}}

\def\shbm{\probfont{SH$\capac$M}}
\def\maxwshbm{\probfont{MaxW}-\shbm{}}
\def\maxwshm{\probfont{MaxW}-\shm{}}
\def\uda{\probfont{UDA}}
\def\tushm{\probfont{Unimod-SHM}}
\def\tushbm{\probfont{Unimod-SH$\capac$M}}

\def\maxwlamshm{\probfont{MaxW}-\lamshm{}}

\def\maxwrcshbm{\probfont{MaxW}-\rcshbm{}}

\newcommand{\maxwsbm}{\probfont{MaxW-S$\capac$M}}

\newcommand{\udamaxwshbm}{\probfont{UDA-MaxW-SH$\capac$M}}

\newcommand{\udashbm}{\probfont{UDA-SH$\capac$M}}

\newcommand{\rcshbm}{\probfont{Subpath-SH$\capac$M}}
\newcommand{\subtrshm}{\probfont{Subtree-SHM}}
\newcommand{\subtrshbm}{\probfont{Subtree-SH$\capac$M}}
\newcommand{\lamshm}{\probfont{Laminar-SHM}}
\newcommand{\lamshbm}{\probfont{Laminar-SH$\capac$M}}

\newcommand{\dist}{{\rm{dist}}}
\newcommand{\mytop}{{\rm{top}}}

\newcommand{\n}{n_S}
\newcommand{\m}{n_U}
\def\EE{\mathcal{E}}

\def\HH{\mathcal{H}}
\def\CC{\mathcal{C}}
\def\NP{\mathsf{NP}}
\def\XP{\mathsf{XP}}
\def\PPAD{\mathsf{PPAD}}
\newcommand{\hh}[1]{\widetilde{#1}}
\newcommand{\worst}{\texttt{worst}}

\newcommand*{\claimproofname}{Proof of Claim.}
\newcommand*{\proofsketchname}{Proof sketch.}
\newenvironment{claimproof}[1][\claimproofname]{\begin{proof}[#1]}{\end{proof}}

\newenvironment{proofsketch}[1][\proofsketchname]{\begin{proof}[#1]}{\end{proof}}

\newcommand{\HE}{\mathcal{E}}
\def\nomatch{\perp}
\def\capac{b}
\def\weight{w}

\makeatletter
\newcommand{\leqnomode}{\tagsleft@true}
\newcommand{\reqnomode}{\tagsleft@false}
\makeatother
\date{}
\begin{document}
\maketitle

\begin{abstract}
We study the $\NP$-hard \textsc{Stable Hypergraph Matching} (\shm) 
problem and its generalization allowing capacities, the \textsc{Stable Hypergraph $\capac$-Matching}  (\shbm) 
problem, and investigate their computational properties under various structural constraints. 
Our study is motivated by the fact that Scarf's Lemma \cite{scarf1967core} together with a result of Lovász \cite{lovasz1972normal} guarantees the existence of a stable matching whenever the underlying hypergraph is normal. Furthermore, if the hypergraph is unimodular (i.e., its incidence matrix is totally unimodular), then 
even a stable  $\capac$-matching is guaranteed to exist.
However, no polynomial-time algorithm is known  for finding a stable matching or $\capac$-matching in unimodular hypergraphs.

We identify subclasses of unimodular hypergraphs where \shm\ and \shbm\ are tractable such as laminar hypergraphs
or so-called  \rcp\ hypergraphs with bounded-size hyperedges; for the latter case, even a maximum-weight stable $\capac$-matching can be found efficiently. 
We complement our algorithms by showing that 
optimizing over stable matchings is $\NP$-hard even in laminar hypergraphs.
As a practically important special case of \shbm\ for unimodular hypergraphs, we investigate a tripartite stable matching problem with students, schools, and companies as agents, called the \textsc{University Dual Admission} 
 problem, which models real-world scenarios in higher education admissions.

Finally, we examine a superclass of \rcp\ hypergraphs that are normal but not necessarily  unimodular, namely subtree hypergraphs where hyperedges correspond to subtrees of a tree.
We establish that for such hypergraphs, stable matchings can be found in polynomial time but, in the setting with capacities, finding a stable $b$-matching is $\NP$-hard.
\end{abstract}

\section{Introduction}

Stable matchings are fundamental in economics, combinatorial optimization, and mechanism design, playing a crucial role in applications such as college admissions, job markets, and organ exchange programs. Since the seminal work of Gale and Shapley~\cite{gale1962college}, stable matchings have been extensively studied in two-sided markets where agents form pairwise relationships. However, many real-world scenarios require interactions beyond pairwise relationships, leading naturally to hypergraph-based formulations. %

The focus of our study is the \probfont{Stable Hypergraph $\capac$-Matching} (\shbm) problem, where we are given a hypergraph whose vertices represent agents, with each agent~$v$ having preferences over the incident hyperedges and, additionally, a capacity value $\capac(v)$. The task is to find a stable $\capac$-matching, i.e., a set~$M$ of hyperedges where each agent~$v$ is adjacent to at most~$\capac(v)$ hyperedges in~$M$, and no hyperedge~$e$ outside~$M$ is ``desirable'' for all agents in~$e$; see Section~\ref{sec:prelim} for the precise definition of stability. 

In the context of hypergraph matchings, stability is often difficult to achieve, as stable solutions may not always exist. On the positive side, Scarf’s Lemma~\cite{scarf1967core} guarantees the existence of a fractional stable solution in the \shbm\ problem. Furthermore, a famous result of Lovász~\cite{lovasz1972normal} in combination with Scarf’s Lemma ensures that in an instance of \textsc{Stable Hypergraph Matching} (\shm)---the restriction of \shbm\ without capacities---with a \emph{normal hypergraph}, a stable matching always exists. An even stronger result follows for \emph{unimodular hypergraphs}, where the incidence matrix is totally unimodular: Scarf’s Lemma guarantees the existence of a stable $\capac$-matching. However, these results only provide existential guarantees, leaving computational aspects largely unexplored.

Motivated by these structural insights, we investigate the computational complexity of finding stable ($\capac$-)matchings in different hypergraph classes. Furthermore, we consider not only the problem of finding \emph{any} stable ($\capac$-)matching but also of computing maximum-weight stable solutions.%

\subsection{Our Contribution}
We begin with a class inspired by a real-world problem: the University Dual Admission (\uda) problem which arises in higher education admissions where students apply to universities and internship programs simultaneously~\cite{fleiner-ferkai-biro2019}. Although we show that a stable solution always exists by reducing \uda\ to \shbm\ in a unimodular hypergraph, our proof does not directly lead to an efficient algorithm. 
though we are unable to settle the complexity of finding a stable solution in a general instance of \uda,
we provide an algorithm that finds a maximum-weight stable solution and runs in $\XP$ time parameterized by the number of universities. 
To complement this, we prove that computing a maximum-size stable matching is $\NP$-hard even %
without capacities, although we show that a stable matching can be found efficiently in this case.
Additionally, we introduce a relaxation of stability, termed \emph{half-stability}, and design a polynomial-time algorithm that always finds a half-stable matching. 

Next, we investigate \emph{laminar hypergraphs}, a class where hyperedges follow a hierarchical inclusion property. Leveraging their structure, we design a polynomial-time algorithm that finds a stable $\capac$-matching whenever the input hypergraph is laminar. By contrast, we also establish that, even in the setting without capacities, finding a maximum-weight stable matching (or even a stable matching containing a fixed edge) is $\NP$-hard.

We extend our study to a superclass of laminar hypergraphs, \emph{\rcp\ hypergraphs}, where hyperedges can be represented as subpaths of a path. Subpath hypergraphs naturally model scenarios where relationships or interactions are constrained along a linear structure, such as transportation networks or supply chains. We develop an $\XP$ algorithm parameterized by the largest hyperedge size for finding a maximum-weight stable $\capac$-matching in such hypergraphs.

Finally, we consider a superclass of \rcp\ hypergraphs, \emph{subtree hypergraphs}, where hyperedges correspond to subtrees of an underlying tree. While these hypergraphs are not necessarily unimodular, they are still normal, which ensures the existence of a stable matching by the results of Lovász~\cite{lovasz1972normal} and Scarf~\cite{scarf1967core}. For the matching case without capacities, we provide a polynomial-time algorithm, which contrasts our result that computing a stable $\capac$-matching for general capacities is $\NP$-hard.

Possible applications that motivate the study of \shbm\ for various hypergraph classes include problems related to budgeting in transportation networks, project selection management in companies, or coalition formation in politics; see Appendix~\ref{sec:app-motivapplic}.

Our findings, summarized in Table~\ref{tab:results}, contribute to the understanding of stable matchings in structured hypergraphs, bridging theoretical guarantees with computational feasibility.

\begin{table}[t]

    \centering
    \renewcommand{\arraystretch}{1.2}
    \setlength{\tabcolsep}{6pt}
    \begin{tabular}{@{}l@{\hspace{8pt}}c@{\hspace{4pt}}c@{\hspace{8pt}}c@{\hspace{4pt}}c@{}}
        \toprule
        & \multicolumn{2}{c}{\textbf{Matching}} & \multicolumn{2}{c}{\textbf{b-Matching}} \\
        & Stable & Max-Weight Stable & Stable & Max-Weight Stable \\
        \midrule
        UDA       & 
        \textbf{P}[R\ref{rem:uda-allone}]& \tnph, \txp($\m$)[T\ref{thm:max-uda},T\ref{thm:constant_uni}]  & \textbf{?}, \txp ($\m$)[T\ref{thm:constant_uni}]  & \tnph, \txp($\m$)[T\ref{thm:max-uda},T\ref{thm:constant_uni}]  \\
        Laminar   & \textbf{P} \cite{chandrasekaran2024scarf} & \tnph, \txp($\ell_{max}$) [T\ref{thm:laminar-weighted-nphard},T\ref{thm:RC1P-fixed-length}] & \textbf{P} [T\ref{thm:laminar-shm-poly}]& \tnph, \txp($\ell_{max}$) [T\ref{thm:laminar-weighted-nphard},T\ref{thm:RC1P-fixed-length}]\\
        Subpath      & \textbf{P} \cite{chandrasekaran2024scarf}  & \tnph, \txp($\ell_{max}$) [T\ref{thm:laminar-weighted-nphard},T\ref{thm:RC1P-fixed-length}] & \textbf{?}, \txp($\ell_{max}$) [T\ref{thm:RC1P-fixed-length}]& \twh, \txp($\ell_{max}$) [T\ref{thm:RC1P-Whardness},T\ref{thm:RC1P-fixed-length}] \\
        Subtree   & \textbf{P}[T\ref{thm:subtree-poly}] & \tnph [T\ref{thm:laminar-weighted-nphard}] & \tnph [T\ref{thm:fa}] & \tnph [T\ref{thm:laminar-weighted-nphard}]\\
        \bottomrule
    \end{tabular}
    \caption{Computational complexity of finding an arbitrary or a maximum-weight stable ($b$-)matching in different hypergraph classes. 
    $\tnph,\textbf{P},\txp$, and $\twh$ refers to $\NP$-hardness,  polynomial-time solvability, solvability in~$\XP$, and $\mathsf{W}[1]$-hardness (the latter two with the given parameter), respectively.
    Parameter $\ell_{max}$ is the size of the largest hyperedge and $\m$ is the number of universities. 
    Recall the hierarchy of hypergraphs which we may write informally as
    $\{\text{laminar}\} \subseteq \{\text{subpath}\} \subseteq \{\text{subtree}\}$.
    \label{tab:results}
    }
    \label{tab:complexity}
\end{table}

\subsection{Related Work}

In their seminal paper, Gale and Shapley \cite{gale1962college} provided a model and an algorithmic solution for college admission problems. 
Their solution concept is called \emph{stability}, which means that an application is fairly rejected by a college if the college filled its quota with better applicants. Gale and Shapley showed that their so-called deferred-acceptance (DA) algorithm always finds a stable solution that is optimal for the students if they are proposing in the algorithm. %

The work of Gale and Shapley~\cite{gale1962college} inspired extensive research in mathematics, computer science, operations research, as well as in economics and game theory. %
For an overview, we recommend the reader to consult the book of Manlove \cite{manlove2013algorithmics} in computer science, the books of Roth and Sotomayor \cite{roth-sotomayor1990} and Haeringer \cite{haeringer2018market} in economics. %
 Besides the theoretical research, the DA algorithm has also been used in many important applications, firstly in the US resident allocation scheme (called NRMP) since 1952~\cite{roth1984} and also in several nationwide college admission and school choice programs. %

One of our motivating problems originates in the Hungarian university scheme where it is possible to apply for a so-called \emph{dual program} that consist of a normal university major together with an internship at a company. For a student to get accepted to such a dual program, both the university and the company have to grant admission, and these parties can have different rankings. Note that this feature is not integrated in the current admission system, as the hiring decisions of the companies are not shared with the central coordinators when conducting the matching algorithm, which causes serious coordination problems, see a detailed description of the problem by Fleiner et al.~\cite{fleiner-ferkai-biro2019}.

Rather independently from the topic of matching under preferences, Scarf~\cite{scarf1967core} showed that so-called balanced NTU-games (a family of cooperative games with non-transferable utility) have non-empty core by providing a finite algorithm that computes a core solution in such games. %
Instead of introducing these game-theoretical notions, we only interpret and use Scarf's Lemma in the context of stable matchings and $\capac$-matchings in hypergraphs.   

The \textsc{Stable Hypergraph Matching} problem, introduced by Aharoni and Fleiner~\cite{AharoniFleiner03}, is essentially equivalent to the problem of finding a core solution in simple NTU-games, also known as \emph{core stability in hedonic coalition formation games}~\cite{banerjee2001core,woeginger2013core,aziz2012existence}; in the hypergraph formulation of these problems, vertices represent agents and hyperedges represent the possible coalitions. %
The setting where agents have capacities, giving rise to the \shbm\ problem, was introduced by Bir\'o and Fleiner \cite{biro2016fractional} who proposed a general framework and proved the existence of fractional stable solutions using Scarf's Lemma for $\capac$-matching problems.
They also highlighted that the matching polytope has the integer property  for normal hypergraphs, therefore the solution determined by Scarf's algorithm for the \shm\ problem is always integer and thus  corresponds to a stable matching.

However, the running time of computing a fractional solution by Scarf's algorithm is not polynomial in general. The PPAD-hardness of computing a fractional stable matching was proved by Kintali et al.~\cite{kintali2013reducibility} and was linked with several other PPAD-hard combinatorial problems. %
Recently Cs\'aji extended these PPAD-hardness results for the fractional stable matching problem for three-partite hypergraphs and for the stable allocation with couples problem \cite{csaji2022complexity}. On the positive side, for certain stable matching problems it was shown that  Scarf's algorithm terminates in polynomial time for some appropriately defined pivot rule. Very recently Faenza et al.\ \cite{faenza2023scarf} showed this property for the classical stable marriage problem, and Chandrasekaran et al.~\cite{chandrasekaran2024scarf} for \shm\ for so-called arborescence hypergraphs (see Section~\ref{sec:prelim} for the definition).
Even though arborescence hypergraphs include subpath hypergraphs, 
the results by Chandrasekaran et al.~\cite{chandrasekaran2024scarf} do not imply any of our results, as they involve neither capacities nor optimization over stable matchings.

\paragraph*{Organization.}
We start by introducing the necessary notation and formally define our studied computational problems in Section~\ref{sec:prelim}. We proceed with the University Dual Admission (\uda) problem in Section~\ref{sec:uda} as a practical example of \shbm. 
We move on to more abstract unimodular hypergraph classes in Section~\ref{sec:algorithms}: in Section~\ref{sec:laminar} we study the class of  laminar hypergraphs, while in Section~\ref{sec:rc1p} we focus on the more general class of \rcp\ hypergraphs.
We examine subtree hypergraphs in Section~\ref{sec:subtree}, and finally conclude in Section~\ref{sec:conclusion}.

Results marked with~($\star$)
have their proofs deferred to the appendices.

\section{Preliminaries}
\label{sec:prelim}
In Section~\ref{sec:notation} we define the concepts related to graphs and hypergraphs that we need in our paper, and in Section~\ref{sec:problems} we provide the definition of the computational problems we will investigate.

\subsection{Notation and Terminology}
\label{sec:notation}
We use the notation $[\ell]=\{1,2,\dots,\ell\}$ for each positive integer~$\ell$. We let $\mathbb{Z}, \mathbb{Z}_+, \mathbb{R}$,  and $\mathbb{R}_+$ denote the set of integers, non-negative integers, reals, and non-negative reals, respectively.
We use basic graph terminology with standard notation; see Appendix~\ref{sec:app-graphs} for all necessary definitions.

\medskip
\noindent
{\bf Matrices, polyhedra and Scarf's Lemma.}
For a matrix~$A$, we let $A_i$ denote its $i$th row.
A matrix~$A$ is
\emph{totally unimodular} or \emph{TU} if each subdeterminant of~$A$ has value in~$\{-1,0,1\}$.
Totally unimodular matrices are known for their useful properties. Most importantly, if $A$ is TU and $\capac$ is integral, then all extreme points of the polyhedron $\{ x\in \mathbb{R}^m:Ax\le b, x\ge 0\}$ are integral. A point $z$ in the polyhedron $\{ x\in \mathbb{R}^m:Ax\le b, x\ge 0\}$
is an \emph{extreme point} if it cannot be written as a strictly convex combination of two different points in the polyhedron. That is, there are no $\lambda \in (0,1)$ and $z_1,z_2\in \{ x\in \mathbb{R}^m:{Ax\le b}, x\ge 0\}$ such that $z=\lambda z_1 + (1-\lambda )z_2$. %

A matrix $A$ is a \textit{network matrix} if it can be obtained from a directed graph $G=(V,E)$ and an edge set~$F \subseteq E$ of~$G$ for which $(V,F)$ is a spanning tree of~$G$ in the undirected sense, via the following method. 
With each row of~$A$ we associate an edge $f\in F$, and with each column we associate an edge $e\in E\setminus F$. 
Given edges~$f \in F$ and~$e \in E \setminus F$, let $C_{e,f}$ denote the unique cycle (in the undirected sense) contained in $(V,F\cup \{ e\})$.
Then, the entry of~$A$ at the intersection of row~$f$ and  column~$e$ has value
\begin{itemize}
    \item $1$, if $f$ is contained in~$C_{e,f}$, and the edges~$e$ and~$f$ have the same orientation along the cycle~$C_{e,f}$; 
\item $-1$, if $f$ is contained in~$C_{e,f}$, and $e$ and~$f$ have a different orientation along~$C_{e,f}$,  and 
\item 
$0$, if $f$ is not contained in the cycle~$C_{e,f}$.
\end{itemize}

It is well known that network matrices are totally unimodular (see e.g., \cite[Theorem 13.20]{Schrijver}).

The following is a key lemma of Scarf \cite{scarf1967core}, useful for many stable matching problems.

\begin{Lemma}[Scarf \cite{scarf1967core}] Let $A \in \mathbb{R}_+^{n \times m}$ be a matrix such that every column of~$A$ has a nonzero element, and let $b \in \mathbb{R}^n_+$. Suppose that every row $i \in [n]$ has a strict ordering $\succ_i$ over those columns $j \in [m]$ for which $A_{ij}>0$. Then there is an extreme point of $\{x \in \mathbb{R}^m: Ax\le b, \; x\ge 0\}$ that dominates every column in some row, where we say that $x\in \mathbb{R}^m$ \emph{dominates} column $j$ in row $i$ if $A_{ij}>0$, $A_ix=b_i$, and $k \succeq_i j$ for all $k\in [m]$ such that $A_{ik}x_k>0$. Also, such an extreme point can be found by a finite algorithm.
\end{Lemma} 

\medskip
\noindent
{\bf Hypergraphs, matchings and $\capac$-matchings.  }
A \emph{hypergraph}~$\HH=(V,\EE)$ contains a set~$V$ of vertices and a set~$\EE \subseteq 2^V$ of \emph{hyperedges}; we may simply refer to hyperedges as \emph{edges} when this causes no confusion. We will say that a hyperedge~$e \in \EE$ is \emph{incident} to a vertex~$v \in V$ if $v \in e$.
The \emph{incidence matrix} of a hypergraph~$\HH=(V,\EE)$ has $|V|$ rows and $|\EE|$ columns, and the entry at the intersection of column~$v \in V$ and hyperedge $e \in \EE$ is~$1$ if $e$ is incident to~$v$, and~$0$ otherwise. For a subset $F\subseteq \EE$, let $F(v) := \{ e\in F\mid v\in e\}$.

Given a hypergraph $\HH = (V,\EE)$, a \emph{matching} is a subset $M\subseteq \EE$ that satisfies that $|M(v)|\le 1$ for every $v\in V$. 
A matching~$M$ leaves a vertex~$v$ \emph{unmatched} if $M(v)=\emptyset$, otherwise it \emph{covers}~$v$.
Given \emph{capacities} $\capac (v)\in \mathbb{Z}_+$ for each $v\in V$, we say that $M\subseteq \EE$ is a \emph{$\capac$-matching} if $|M(v)|\le \capac (v)$ for each $v\in V$.
For a $\capac$-matching $M$, we say that $v$ is \textit{unsaturated} in~$M$ if $|M(v)|<\capac (v)$, 
 \textit{saturated} if 
$|M(v)|=\capac (v)$, and 
\textit{oversaturated} 
if $|M(v)|>\capac (v)$.

\medskip
\noindent
{\bf Hypergraph classes.  }
Let us now describe the hypergraph classes relevant to our paper.

Let $\HH$ be a hypergraph, and let $\HH'$ be any \emph{partial hypergraph} of $\HH$, obtained by deleting some of its edges.  
The \emph{chromatic index} of a hypergraph, denoted $\chi_e(\HH)$, is the smallest number of colors required to color its edges  
such that no two edges of the same color share a common vertex. It is clear that the maximum degree  
$\Delta(\HH)$—the largest number of edges incident to any single vertex—provides a lower bound for $\chi_e(\HH)$.  

\begin{Def}A hypergraph $\HH$ is \emph{normal} if $\chi_e(\HH') = \Delta(\HH')$ for every partial hypergraph $\HH'$ of $\HH$.
\end{Def}
Lovász \cite{lovasz1972normal} provided an alternative characterization of normal hypergraphs.%

\begin{theorem}[Lovász]
\label{thm:lovasz}
    A hypergraph $\HH$ is normal if and only if all extreme points of the polyhedron $\{ x\in \mathbb{R}^m: Ax\le 1, x\ge 0\}$ are integral, where $A$ is the incidence matrix of $\HH$.
\end{theorem}

An important subclass of normal hypergraphs is the class of unimodular hypergraphs.
\begin{Def}
A hypergraph is \emph{unimodular} if its incidence matrix is totally unimodular.
\end{Def}
Next, we introduce subclasses of unimodular hypergraphs.
\begin{itemize}
    \item A \emph{network hypergraph} is a hypergraph whose incidence matrix is a network matrix.
    \item \emph{Arborescence hypergraphs} are a subclass of network hypergraphs:
$\HH = (V,\EE)$ is an \emph{arborescence hypergraph} if there is an arborescence $F$ over~$V$  such that each edge of~$\EE$ forms a directed path in~$F$.
\item \emph{Subpath hypergraphs} are the special case of arborescence hypergraphs where the underlying arborescence is required to be a path.
    Formally, $\HH = (V,\EE)$ is a \emph{\rcp\ hypergraph} if there exists a directed path $F$ over~$V$  such that each edge of~$\EE$ forms a  subpath in~$F$.
    \item %
A \emph{laminar hypergraph} is a hypergraph~$\HH=(V,\EE)$ for which there are no two hyperedges~$e_1,e_2 \in \EE$ that satisfy $e_1 \setminus e_2 \neq \emptyset$ and $e_2 \setminus e_1 \neq \emptyset$. 
Laminar hypergraphs are known to be \rcp\ hypergraphs; although this seems to be folklore, we provide a proof in Appendix~\ref{sec:app-hyp}.
\end{itemize}

Finally, we define a superclass of arborescence hypergraphs.
A hypergraph $\HH=(V,\EE)$ is a \textit{subtree hypergraph} if there is a tree $F=(V,E)$ on the ground set~$V$ such that each hyperedge induces a subtree of $F$. 
While subtree hypergraphs are not necessarily unimodular, they still form a subclass of normal hypergraphs \cite[Section 4.4]{bretto-hypergraphs}. 

Appendix~\ref{sec:app-hyp} offers a full, graphical overview of the inclusion relations among these hypergraph classes.

\subsection{Stable Hypergraph Matching}
\label{sec:problems}
We start by defining the \textsc{Stable Hypergraph Matching} and \textsc{Stable Hypergraph $\capac$-Matching} problems. Let $\HH=(V,\EE )$ be a hypergraph. For each $v\in V$  we are given a strict preference list~$\succ_v$ over the  set $\EE(v)$ of  hyperedges containing~$v$. In the case of \textsc{Stable Hypergraph $\capac$-Matching}, we are also given capacities $\capac(v)\in \mathbb{Z}_+$ for each $v\in V$.

\begin{Def}
\label{def:shm-stability}
    A hyperedge~$e\in \EE$ is \emph{dominated} by a ($\capac$-)matching~$M$ at some vertex~$v \in V$ if $v$ is saturated in~$M$ and
    $f \succ_v e$ for each hyperedge $f \in M(v)$. 
    A hyperedge $f$ \emph{blocks}~$M$ if $f\notin M$ and there exists no vertex~$v\in f$ at which  $M$ dominates~$f$. 
 A ($\capac$-)matching is \textit{stable} if no hyperedge blocks it.
\end{Def} 

The most important question in such instances is to find a stable matching or $\capac$-matching. Hence, we get the following computational problem.

\searchprob{Stable Hypergraph $\capac$-Matching \textup{or} \shbm{}}
{
A hypergraph $\HH=(V,\EE)$ with capacities $\capac:V \rightarrow \mathbb{Z}$ and strict preferences $(\succ_v)_{v \in V}$ over incident edges for each vertex in~$V$.
}
{
Find a stable hypergraph $\capac$-matching in~$\HH$.
}

\textsc{Stable Hypergraph Matching} (or \shm) is defined analogously for matchings instead of $\capac$-matchings. It is clear that \shm\ is a subcase of \shbm\ where we have $\capac \equiv 1$.

For $\Pi \in \{$\textsc{Unimod}, \textsc{Subpath},  \textsc{Subtree}$\}$, we denote by $\Pi$-\shm\ and $\Pi$-\shbm\ the restrictions of \shm\ and \shbm\ to unimodular; subpath; or subtree hypergraphs, respectively.

As was observed by Aharoni and Fleiner~\cite{AharoniFleiner03}, every stable hypergraph matching for an instance $(\HH ,(\succ_v)_{v \in V})$ of \shm\  with hypergraph~$\HH=(V,\EE)$
can be seen as an integral point
in the polyhedron $\{x \in \mathbb{R}^m: Ax \leq 1, x \geq 0\}$ that dominates every column according to the preferences~$(\succ_v)_{v \in V}$, where $A$ is the incidence matrix of~$\HH$. Thus, Scarf's Lemma and Theorem~\ref{thm:lovasz} imply the existence of a stable matching if the hypergraph is normal.
Furthermore, for an instance  $(\HH,\capac ,(\succ_v)_{v \in V})$ of \shbm, the integral points of the polyhedron $\{ x\in \mathbb{R}^m: Ax\le \capac, 0\le x\le 1\}$ that dominate every column correspond to the stable hypergraph $\capac$-matchings, where $\capac$ (with a slight abuse of notation) is the vector containing the capacities of the vertices.
If the incidence matrix~$A$ is TU, then all extreme points of this polyhedron are integral, and thus Scarf's Lemma guarantees the existence of a stable $\capac$-matching whenever the underlying hypergraph is unimodular. 
The guaranteed existence of a stable matching or stable $\capac$-matching makes the class of normal hypergraphs and its subclass, unimodular hypergraphs, particularly interesting and useful.%

Although the existence of a stable matching or $b$-matching may be guaranteed in restricted hypergraph classes, it is also of interest to find a maximum-weight stable hypergraph matching.
For this, an optimization variant of \shbm\ (and \shm) can be formulated by associating a \emph{weight} with each hyperedge.
Given a weight function $\weight:\EE \rightarrow \mathbb{Z}$ over a set~$\EE$ of hyperedges, we extend it   by setting $\weight(M)=\sum_{e \in M}\weight(e)$ for each $M \subseteq \EE$.

\searchprob{Maximum-Weight \shbm\ \textup{or} \maxwshbm{}}
{
A hypergraph $\HH=(V,\EE )$ with capacities $\capac:V \rightarrow \mathbb{Z}$, weights $\weight:\EE \rightarrow \mathbb{Z}$, and strict preferences $(\succ_v)_{v \in V}$ over incident edges for each vertex in~$V$.
}
{
Find a maximum-weight stable hypergraph $\capac$-matching in~$\HH$.
}

We will refer to the optimization variant of the matching case, i.e., \shm\ as
\textsc{Maximum-Weight \shm{}} or \maxwshm.
For $\Pi \in \{$\textsc{Unimod}, \textsc{Subpath},  \textsc{Subtree}$\}$, we denote by $\Pi$-\maxwshm\ and $\Pi$-\maxwshbm\ the restrictions of \maxwshm\ and \maxwshbm\ to unimodular, subpath, or subtree hypergraphs, respectively.

We remark that in the special case when the hypergraph $\HH$ is a bipartite graph, \maxwshbm\  coincides with the maximum-weight stable $\capac$-matching problem (\maxwsbm) which is solvable in polynomial time~\cite{Fleiner03-bmatching}.

\section{The University Dual Admission Problem}\label{sec:uda}

In this section, we introduce a practically relevant problem in student allocation which, as we will show, can also be modeled as a special case of \tushm.

The problem is motivated by practical situations where certain programs in higher education are funded jointly as a result of cooperation between universities and companies, leading to a more complicated dual admission system. In our model, motivated by a real scenario in Hungarian higher education, we assume that such programs, funded by various companies, can be treated independently from each other (even though different programs might be funded by the same company)---hence, we shift the focus from the funding companies to the programs themselves.

Formally, we have a set $U=\{ u_1,\dots,u_{\m} \}$ of \emph{universities}, a set $P_i=\{ p_{i1},\dots,p_{ik_i}\}$ of \emph{programs} for each university~$u_i \in U$, and a set $S=\{ s_1,\dots,s_{\n}\}$ of \emph{students}.  The set of programs is denoted by~$P=\bigcup_{i \in [\m]} P_i$. We may refer to the set of students, universities, and programs together as the set of \emph{agents}.
Each university~$u_i\in U$ has a \emph{capacity}~$\capacuda (u_i)$, and each program~$p_{ik} \in P$ has a \emph{quota}~$\capacquota (p_{ik})$.

The students apply to both a university $u_i \in U$ and a program $p_{ik} \in P_i$  available at university~$u_i$. 
Since each program~$p_{ik} \in P$ uniquely determines the university~$u_i$ offering it, this can alternatively be viewed as students applying simply to programs.\footnote{Nevertheless, universities play an important role in the admission problem, as will become clear in the definition of a stable assignment (Definition~\ref{def:uda-stability}).}
Thus, for each student $s_j \in S$ we assume a strict preference order $\succ_{s_j}$ over the set of programs acceptable for~$s_j$. 
Additionally, each university~$u_i \in U$ has a strict preference order~$\succ_{u_i}$ over the students, and each program $p_{ik} \in P$ has a strict ordering~$\succ_{p_{ik}}$ over the students.
A triple $(s_j,u_i,p_{ik})$ is an \emph{acceptable triple} if $s_j$ finds $p_{ik}$ acceptable; note that we do not consider acceptability for universities or programs explicitly, as we implicitly assume that a student finds a program acceptable only if both the program and university offering it finds the student acceptable as well.
We may treat acceptable triples as subsets of~$S \cup U \cup P$ of size~$3$.

An \emph{assignment} $\M$ is a function~$\M:S \rightarrow P \cup \{\nomatch\}$ mapping each student ${s_j \in S}$ either to a program that is acceptable for~$s_j$, or to the special symbol~$\nomatch$ meaning that~$s_j$ is left unassigned. For simplicity, we assume that a program is acceptable for~$s_j$ if and only if it is preferred by~$s_j$ to~$\nomatch$. For a program~$p_{ik} \in P$, we let $\M(p_{ik})$ denote the set of students assigned to it by~$\M$, and similarly, for a university~$u_i \in U$ we let $\M(u_i)$ denote the set of students assigned by~$\M$ to some program in~$P_i$. An assignment~$\M$ is \emph{feasible} if $|\M(u_i)|\le \capacuda (u_i)$ for each $u_i \in U$ and $|\M(p_{ik})|\le \capacquota(p_{ik})$ for each $p_{ik} \in P$; we say that a university or a program is \emph{unsaturated} if the corresponding inequality is strict, otherwise it is \emph{saturated}.

\begin{Def}
\label{def:uda-stability}
A feasible assignment~$\M$ for an instance of \uda\ is \emph{stable} if there is no \emph{blocking} student--university--program triple, that is, a triple $(s_j,u_i,p_{ik})$  such that the following three conditions hold:
\begin{itemize}
    \item[(i)] $p_{ik}\succ_{s_j} \M(s_j)$; 
    \item[(ii)] $u_i$ satisfies one of the followings: 
    \begin{itemize}
        \item $u_i$ is unsaturated,
        \item there is a student $s_{j'}\in \M(u_i)$ such that $s_j\succ_{u_i}s_{j'}$, or
        \item   $s_j\in \M(u_i)$;
    \end{itemize}  
    \item[(iii)] the program $p_{ik}$ is unsaturated, or there is a student $s_{j''}\in \M(p_{ik})$ such that $s_j\succ_{p_{ik}}s_{j''}$.
\end{itemize}    
\end{Def}
We are ready to formally define the corresponding computational problem.
\decprob{University Dual Admission \textup{or} UDA}
{
An instance $I=(S,U,P,\capacuda,\capacquota,(\succ_a)_{a \in S \cup U \cup P})$ of \uda\ as described above.
}
{
Find a feasible and stable assignment for $I$.
}

Let us remark that the above problem can easily be used to model a situation where students may apply either to university--program pairs or only to universities:
to allow for this option, we simply need to create a dummy program~$p^*_i$ for each university~$u_i \in U$ with quota~$\capacuda(u_i)$ whose preferences are identical to the preferences of~$u_i$.

\smallskip
A very natural idea to solve  \uda\ is to use a Gale--Shapley-like  proposal--rejection algorithm and iterate it until it outputs a stable matching. We give an example where a student-proposing variant of the Gale--Shapley  algorithm  goes into an infinite loop when used on an instance of~\uda, even if students propose in a lexicographic order. It is also possible to create instances where universities or programs propose but the algorithm still cycles in a similar way. This suggests that the problem should be solved in a different manner.

\subsection{\texorpdfstring{Formulating \uda\ as a Special Case of \shbm}{Formulating UDA as a Special Case of SHbM}}
\label{sec:uda-vs-shm}
In this section, %
we first show that \uda\ can be formulated as a special case of the \tushbm, and as a consequence, 
a stable assignment always exists for every instance of~\uda. 
By contrast, we also prove that a stable assignment of maximum size is $\NP$-hard to find.

\begin{Def}
\label{def:uda-hypergraph}
Given an instance~$I$ of \uda, let $\HH_I=(S \cup U \cup P,\EE_I)$ denote the
\emph{hypergraph associated with~$I$} whose vertices are the agents in~$I$ and whose hyperedges are the acceptable triples in~$I$. %
We say that a hypergraph $\HH$ has the \emph{\uda\ property} if it can arise as $\HH_I$ for some instance $I$ of \uda.
\end{Def}

We call \udashbm\ and \udamaxwshbm\ the restriction of \shbm\ and \maxwshbm\ to hypergraphs having the \uda\ property, respectively.
 
To create an instance of \shbm\ that is equivalent with our instance~$I$ of \uda, we also need to define a strict ranking over the hyperedges of~$\HH_I$ incident to each  student, university, or program. For a student~$s_j \in S$, his preference list~$\succ_{s_j}$ in~$I$ %
can be viewed as a preference list over its incident hyperedges: $s_j$ prefers a hyperedge $\{ s_j,u_i,p_{ik}\}$ over~$\{ s_j,u_{i'},p_{i'k'}\}$ if and only if $p_{ik} \succ_{s_j} p_{i'k'}$. For a program~$p_{ik} \in P$, the preference order $\succ_{p_{ik}}$ %
similarly defines a ranking over the incident hyperedges, as each such hyperedge contains the same university~$u_i$. For a university~$u_i \in U$, we extend its preference list to all incident hyperedges the following way. We let $u_i$ prefer $\{ s_j,u_i,p_{ik}\}$ to $\{ s_{j'},u_i,p_{ik'}\} $ if and only if either $s_j \succ_{u_i}s_{j'}$ or $s_j=s_{j'}$ and $p_{ik} \succ_{s_j} p_{ik'}$. Let $\succ_a^{\HH}$ denote the preferences of agent~$a \in S \cup U \cup P$ in~$I$ thus defined.

Additionally, we define the capacity function $\capac^{\HH}:S \cup U \cup P$ over the vertices of~$\HH_I$ as 
\[\capac^{\HH}(a)=
\left\{
\begin{array}{ll}
1, & \textrm{ if $a \in S$,} \\
\capacuda(u_i),  & \textrm{ if $a =u_i \in U$,} \\
\capacquota(p_{ik}), & \textrm{ if $a=p_{ik} \in P$.}
\end{array}
\right.
\]

\begin{restatable}[\linkproof{sec:app-proof-uda-reduces-to-shm}]{theorem}{thmudareducestoshm}
\label{thm:uda-reduces-to-shm}
    The stable assignments of an instance of~$I=(S,U,P,\capacuda,\capacquota,(\succ_a)_{a \in S \cup U \cup P})$ of \uda\ %
    correspond bijectively to
    the stable $\capac^{\HH}$-matchings of the instance $I'=(\HH_I,\capac^{\HH},{(\succ^{\HH}_a)_{a \in S \cup U \cup P})}$ of \shbm.
\end{restatable}

Next, we show that the incidence matrix of the hypergraph associated with an arbitrary instance of \uda\ is a network hypergraph and is thus unimodular.

\begin{theorem}
\label{thm:uda-network-shm}
For each instance~$I$ of \uda, the associated hypergraph~$\HH_I$ is a network hypergraph. %
\end{theorem}
\begin{proof}
We create a directed graph~$D=(V,E)$ together with a spanning tree $F$ in the undirected graph $\overline{D}=(V,\overline{E})$ obtained by replacing each arc~$(u,v)$ of~$D$ by an edge between $u$~and~$v$. We will show that the incidence matrix of~$\HH_I$ is exactly the network matrix that corresponds to~$D$ with the spanning tree~$F$.

We create a vertex in~$V$ for each student, university, and program, and we further define an additional vertex~$x$. Formally, we set \[V=\{v(a):a \in S \cup U \cup P\} \cup \{x\}.\] 
The set~$E$ of  arcs is defined as follows. First, for each program $p_{ik}$ at some university~$u_i$, we add an arc $e(p_{ik})=(v(p_{ik}),v(u_i))$. Second, for each university~$u_i$, we add an arc $e(u_i)=(v(u_i),x)$. 
Third, for each student $s_j$, we add an arc $e(s_j)=(x,v(s_j))$. 
Let $E_F=\{e(a):a \in S \cup U \cup P\}$ denote the set of arcs defined so far.
Finally, for each acceptable triple $(s_j,u_i,p_{ik})$, i.e., for each hyperedge in~$\HH_I$, we add an arc $(v(s_j),v(p_{ik}))$. 
It is straightforward to see that $F:=\{\{u,v\}:(u,v) \in E_F\}$ is a spanning tree in~$\overline{D}$.
Let $A$ be the network matrix associated with $D$ and~$F$.
Note that the edges of~$F$ are in a one-to-one correspondence with the vertices of the hypergraph $\HH_I$, i.e., the set of agents in~$I$. Furthermore, the edges of $E\setminus F$ are in a one-to-one correspondence with the hyperedges of $\HH_I$.

Thus, it only remains to show that $A$ is indeed the incidence matrix of~$\HH_I$, i.e., if a vertex $v\in V$ is contained in a hyperedge $\hat{e} \in \EE_I$, then the corresponding entry in~$A$ is~1, and~0 otherwise. To see this, consider the unique cycle~$C_{\hat{e}}$ defined by an arc $(v(s_j),v(p_{ik}))$ corresponding to some hyperedge $\hat{e}=\{ s_j,u_i,p_{ik}\}$, which is the cycle formed by the four edges $ (v(s_j),v(p_{il}))$, $e(p_{ik})$, $e(u_i)$, and~$e(s_j)$. Observe that $C_{\hat{e}}$ contains exactly three arcs of~$F$, namely the arcs of~$D$ corresponding to the three vertices (i.e., agents) incident to~$\hat{e}$. 
Hence, there are exactly three non-zero entries in the  column of~$A$ corresponding to~$\hat{e}$, and these three entries are in the three rows that correspond to the vertices incident to~$\hat{e}$.
Moreover, since all arcs of~$C_{\hat{e}}$ have the same orientation along~$C_{\hat{e}}$, each of these three entries of~$A$ has value~$1$, as required.
\end{proof}

While Theorems~\ref{thm:uda-reduces-to-shm} and~\ref{thm:uda-network-shm} guarantee the existence of a stable assignment for every \uda\ instance, a maximum-size stable assignment is $\NP$-hard to find, where the \emph{size} of an assignment is the number of students it assigns  to some program.

\begin{restatable}[\linkproof{sec:app-proof-max-uda}]{theorem}{thmmaxuda}
\label{thm:max-uda}
Finding a maximum size stable assignment in an instance of \uda\ is $\NP$-hard, even if all capacities are one.
\end{restatable}

To handle this computational difficulty, we formulate an integer program whose solutions are exactly the stable assignments  for an instance~$I=(S,U,P,\capacuda,\capacquota,(\succ_a)_{a \in S \cup U \cup P})$ of \uda; see Appendix~\ref{sec:app-uda-ILP}.

\subsection{\texorpdfstring{An XP Algorithm for \udamaxwshbm}{An XP Algorithm for UDA-MaxW-SHbM}}

Let us describe an algorithm for an instance
$I=(\HH,\capac,\weight,(\succ_v)_{v \in V})$ of \udamaxwshbm, where $\HH = (V,\EE)$ over vertex set $V=S\cupdot U\cupdot P$ is a hypergraph that satisfies the \uda\ property.  Let $U=\{u_1,\dots,u_{\m}\}$. We can assume that $P=\bigcupdot_{i\in [\m]}P_i$, where the vertices $p\in P_i$ satisfy that for any $p\in e\in \EE$, we have $\{ u_i,p\} \subset e$.

The pseudocode of the algorithm is given in Algorithm~\ref{alg:uda-xp}.
As a first step, we guess a \emph{strategy} $\sigma : U\to \EE \cup \{ \nomatch \}$, that is, for each $u_i\in U$ we guess whether it will be saturated in the solution, and if so, what will be the worst hyperedge adjacent to $u_i$. 
Then, we create an instance $I'_{\sigma}$ of \maxwsbm, the maximum-weight bipartite stable $\capac$-matching problem.
The underlying graph of $I'_{\sigma}$ will be the bipartite graph $G_\sigma=(S \cup P,E_{\sigma})$ where $E_\sigma$ contains an edge~$\{s,p\}$ if (i) there is a (unique) vertex $u\in U$ such that $e=\{ s,u,p\} \in \EE$ and (ii) for that vertex $u$, $e \succeq_{u} \sigma(u)$, where we assume that $e \succeq_{u} \nomatch$ always holds.  
The preferences of the agents in $I'_{\sigma}$ are the projection of their original preferences. This is well defined, as each edge~$\{s,p\} \in E_\sigma$ uniquely defines a hyperedge $\{ s,u,p\} \in \EE$; 
we further set $\capac '(s)=\capac (s)$, $\capac'(p)=\capac(p)$, 
and $\weight' (\{ s,p\}) = \weight (\{ s,u,p\})$ so that 
capacities and weights remain the same.

Let us say that a $\capac'$-matching $M'$ in~$I'_{\sigma}$ is \emph{valid} if, $M:= \{ \{ s,u,p\}\in \EE\mid \{ s,p\} \in M'\}$ satisfies that for each $u\in U$, $|M(u)|=\capac (u)$ whenever $\sigma(u) \neq \nomatch$, and 
$ |M(u)|<\capac (u)$ otherwise. 

For a given strategy~$\sigma$, the algorithm proceeds by computing a maximum-weight stable $\capac'$-matching~$M'$ in~$I'_{\sigma}$. If $M'$ is valid, then it adds the $\capac$-matching~$M:= \{ \{ s,u,p\}\in \EE\mid \{ s,p\} \in M'\}$ corresponding to~$M'$ to a set $\stabset$. If $M'$ is not valid, then we discard the guessed strategy~$\sigma$. Finally, we output a maximum-weight stable $\capac$-matching from $\stabset$.

\begin{algorithm}
\caption{\udamaxwshbm($\HH=(V,\EE),\capac,\weight,(\succ_v)_{v \in V})$ where $\HH$ defined over $S \mycupdot U \mycupdot P$ has the \uda\ property}
\label{alg:uda-xp}
\begin{algorithmic}[1]
\State $\stabset:= \emptyset$
\For{all strategies $\sigma : U\to \EE\cup \{ \nomatch \}$
}
\State $E_{\sigma}:=\emptyset$
\For{each $e=\{s,u,p\} \in \EE$}
\If{$e\succeq_{u}\sigma (u)$} %
$E_{\sigma}:= E_{\sigma}\cup \{ \{s,p\}\}$
\EndIf
\EndFor
\State Set $I'_{\sigma} := ((S\cup P, E_{\sigma}), \capac', \weight',(\succ'_v)_{v\in S\cup P} )$ where 
$b'$, $\weight'$, and $\succ'_v$ are obtained by projection from $b$, $~\weight$, and $\succ_v$, respectively.
\State %
Compute a maximum-weight stable $\capac'$-matching $M'$ in $I'_{\sigma}$.
\State  $M:=\{ \{s,u,p\} \in \EE\mid \{ s,p\} \in M'\} $
\If{$|M(u)| = \capac (u)$ whenever $\sigma (u)\ne \nomatch$, and $|M(u)|<\capac (u)$ otherwise for all $u\in U$}
\State $\stabset:= \stabset\cup \{ M\}$
\EndIf
\EndFor
\State \textbf{return} a maximum-weight $\capac$-matching $M$ from $\stabset$.
\end{algorithmic}

\end{algorithm}

\begin{theorem}
    \label{thm:constant_uni}
  \udamaxwshbm\ can be solved in $\mathcal{O}((\Delta_U+1)^{|U|})\cdot \poly( |\EE|)$ time, where $U$ correspond to the class of universities and $\Delta_U$ is the maximum degree of a vertex $u\in U$.  
\end{theorem}

\begin{proof}

We show that in the end of Algorithm~\ref{alg:uda-xp}, $\stabset$ will contain a maximum-weight stable $\capac$-matching for~$I$. 

\begin{claim}
\label{clm:const-uni}
    If $M'$ is a valid stable $\capac'$-matching in~$I'_{\sigma}$, then the corresponding $\capac$-matching $M=\{ {\{s,u,p\}  \in \EE} \mid \{ s,p\} \in M'\} $ is feasible and stable in~$I$. 
\end{claim}
\begin{claimproof}
Let $M'$ be a valid stable $\capac'$-matching in~$I'_\sigma$ and let $M$ be defined as in the statement of the claim. Then, $M$ is a feasible $\capac$-matching in~$I$: for $v\in S\cup P$, $|M(v)|\le \capac (v)$ follows from $\capac '(v)=\capac (v)$, and the capacities of the agents $u\in U$ are respected due to the validity of~$M'$. It remains to show that $M$ is stable in~$I$.

For a hyperedge $e=\{ s,u,p\}$, let $e_{\spr}=\{ s,p\}$.
Suppose on the contrary that a hyperedge $e=\{ s,u,p\} $ blocks~$M$. 
First, if $e_{\spr} \in E_\sigma$, then as $p$ (respectively, $s$) is either unsaturated in both~$M$ and~$M'$, or $e\succ_p f$ ($e\succ_s f'$) for some $f\in M(p)$ ($f'\in M(s)$) implying $e_{\spr}\succ'_pf_{\spr}\in M'$ ($e_{\spr}\succ'_sf'_{\spr}\in M'$); 
 this means that the edge $e_{\spr}$ blocks~$M'$, a contradiction. Second, if $\{s,p\}\notin E_\sigma$, then $\sigma(u) \succ_u e$. In particular, $\sigma (u)\ne \nomatch$ and due to the validity of~$M'$ we know that $u$ is saturated in~$M$. 
Moreover, by construction, any edge~$\{s,p\}\in M'$ satisfies $\{ s,u,p\} \succeq_{u} \sigma (u)$. Therefore, $u$ is saturated in~$M$ with hyperedges that are not worse for~$u$ than~$\sigma (u)$.
Hence $e$ cannot block~$M$, a contradiction.
\end{claimproof}

\begin{claim}
\label{clm:constuni-rightguess}
    Suppose that $M$ is a stable $\capac$-matching in~$I$. Define $\sigma:U \rightarrow \EE \cup \{\nomatch\}$ such that $\sigma(u)$ is the worst hyperedge in~$M(u)$ according to~$\succ_u$ if $u$ is saturated in~$M$, and $\nomatch$ otherwise. Then every stable $\capac'$-matching $M'$ in $I'_{\sigma}$ is valid.
\end{claim}
\begin{claimproof}
Let $M'$ be the $\capac'$-matching in $I'_{\sigma}$ defined as $M'=\{ \{s,p\}:\{s,u,p\}\in M\}$. By the definition of $\sigma$, it holds that $M' \subseteq E_{\sigma}$, and moreover, $M'$ is valid by construction. We claim that $M'$ is stable in~$I'_{\sigma}$. 

To show that $M'$ is stable in $I'_{\sigma}$, assume for the sake of contradiction that some edge $\{s,p\} \in E_\sigma$ blocks $M'$. Then, both $s$ and $p$ are either unsaturated or prefer $\{ s,p\}$ to some edge in~$M'$. 
Since $M$ is stable, this means that the hyperedge $e=\{s,u,p\}$ must be dominated at $u$. Therefore, we get that $u$ is saturated in $M$ and $\sigma (u)\succ_u e$. However, then the edge~$\{s,p\}$ cannot be contained in~$E_\sigma$, a contradiction.

Due to the many-to-many extension of the well-known Rural Hospitals Theorem (see e.g., \cite{Fleiner03-bmatching}), for any stable $\capac'$-matching $M''$ in~$I'_{\sigma}$, we have $|M''(v)|=|M'(v)|$ for all $v\in S\cup P$. Hence, all stable $\capac'$-matchings of $I'_{\sigma}$ must be valid, because for any stable $\capac'$-matching $N'$ and its corresponding $\capac$-matching~$N$ in~$I$, we have $|N(u_i)|=\sum_{p\in P_i}|N'(p)|$ for $u_i\in U$ by the \uda\ property. This concludes the proof.
\end{claimproof}

As we know that a stable $\capac$-matching for~$I$ exists, Claims \ref{clm:constuni-rightguess} guarantees that 
by iterating over all possible strategies, the algorithm is bound to find the strategy $\sigma$ corresponding to the maximum-weight stable $\capac$-matching $M$. By Claims~\ref{clm:const-uni} and \ref{clm:constuni-rightguess},  $M'=\{ \{ s,p\} \mid \{ s,u,p\} \in M\}$ will be a valid stable $\capac'$-matching for $I_{\sigma}$, and since any valid $\capac'$-matching $N'$ gives a stable $\capac$-matching $N$ in $I$ with $\weight (N')=\weight (N)$, the algorithm indeed puts $M$ or a stable $\capac$-matching with the same weight in $\stabset$.

Since there are at most $(\Delta_U+1)^{|U|}$ possible strategies to check, and for each one them, the instance~$I'_\sigma$ and a maximum-weight stable $\capac$-matching in~$I'_\sigma$ can be computed in $\poly(|\EE|)$ time by a result of Fleiner~\cite{Fleiner03-bmatching}, the claimed running time also follows. \end{proof}

\paragraph{Further results.} On the positive side, we define a relaxed notion of stability in \uda, called \emph{half-stability}, which intuitively means that for a triple to block, if the program is saturated, then the university and the program should be able to agree on a student that can be rejected; we show that we can find a half-stable assignment in polynomial time.  Due to space constraints, this result is deferred to Appendix~\ref{sec:half-stab}. 
\begin{Rem}\label{rem:uda-allone}
If $\capacuda (u)=1$ for all universities, then a simple deferred acceptance algorithm between the students and universities produces a stable assignment. This statement is explained further and proved formally in Appendix \ref{sec:uda-allone}. 
\end{Rem}
\section{\texorpdfstring{\shbm\ for Unimodular Hypergraphs}{SHbM for Unimodular Hypergraphs}}
\label{sec:algorithms}

We proceed to study \tushbm\ in further classes of unimodular hypergraphs.

\subsection{Laminar Hypergraphs}
\label{sec:laminar}
We construct a polynomial-time algorithm that finds a stable $\capac$-matching if the hypergraph $\HH=(V,\mathcal{E)}$ is laminar. As laminar hypergraphs are unimodular, the existence of a stable $\capac$-matching is guaranteed.
We present our algorithm for this problem in Algorithm~\ref{alg:laminar}.

\begin{algorithm}
\caption{Laminar-\shbm%
$(\HH,\capac,(\succ_v)_{v \in V})$ where $\HH=(V,\EE)$ is  laminar}
\label{alg:laminar}
\begin{algorithmic}[1]
\State $M:=\emptyset$  \Comment{$M$ is the current $\capac$-matching}
\State $\mathcal{C}:=\emptyset$   \Comment{$\mathcal{C}$ denotes the checked hyperedges}
\While{$\EE\setminus \mathcal{C} \ne \emptyset$}
\State Pick an inclusion-wise minimal hyperedge $f$ in $\EE\setminus \mathcal{C}$.\label{line:pickedge}
\State $\mathcal{C}:=\mathcal{C}\cup \{f\}$

\If{$f$ blocks~$M$}

    \State Set $M:=M\cup \{f\}$\label{line:addedge} and initialize $\remsetL_f:=\emptyset$.
    \ForAll{$v \in V$ with $|M(v)|>\capac(v)$}
        \State Add to~$\remsetL_f$ the hyperedge in~$M(v)$ that is worst for~$v$.\label{line:compute-Df}
    \EndFor
    \State Set $\remsetFIN_f:=\{e:e \textrm{ is inclusion-wise maximal in } \remsetL_f\}$ \label{line:compute-Dmaxset} and $M:=M \setminus \remsetFIN_f$\label{line:adjustM}.
\EndIf \label{line:ifend}
\EndWhile
\State \textbf{return} $M$
\end{algorithmic}
\end{algorithm}

\begin{theorem}
\label{thm:laminar-shm-poly}
\lamshbm\ is solved in polynomial time by Algorithm~\ref{alg:laminar}. 
\end{theorem}
\begin{proof}
First, let us show that the set returned by Algorithm~\ref{alg:laminar} is a $\capac$-matching.
Observe that whenever some hyperedge~$f$ is added to~$M$ on line~\ref{line:addedge}, the set~$\remsetL_f$ constructed on line~\ref{line:compute-Df} covers every vertex~$v$ that became oversaturated as a consequence of adding~$f$. Therefore, $\remsetFIN_f$ also covers these vertices, so  
after the removal of all hyperedges of~$\remsetFIN_f$ from the matching, the resulting set again satisfies all capacity constraints. Therefore, the algorithm indeed returns a  $\capac$-matching~$M$.

It remains to show that $M$ is stable. Consider the iteration of lines~\ref{line:pickedge}--\ref{line:ifend} during Algorithm~\ref{alg:laminar} where some hyperedge~$f$ is picked on line~\ref{line:pickedge}, and let  $M_f$ and $M'_f$ denote the current $\capac$-matching at the beginning and at the end of this iteration, respectively.
We need the following claim. 
\begin{claim}
    For any edge~$f$ added to the $\capac$-matching at line~\ref{line:addedge}, 
    the set $\remsetFIN_f$ is a family of pairwise disjoint hyperedges  whose union is a subset of~$f$, with $f \notin \remsetFIN_f$.
\end{claim}
\begin{claimproof}
Since the algorithm always picks an inclusion-wise minimal hyperedge on line~\ref{line:pickedge}, the laminarity of~$\HH$ implies that every hyperedge that had been checked (i.e., put into~$\CC$) before $f$ is either a subset of~$f$ or is disjoint from it. In particular, since the algorithm only puts hyperedges of the current $\capac$-matching into~$\remsetL_f$, we get that all hyperedges in~$\remsetL_f$ are contained in~$f$ 
(recall that $f$ shares at least one oversaturated vertex with each hyperedge in~$\remsetL_f$). 
Since $\remsetFIN_f$ contains only the inclusion-wise maximal sets from~$\remsetL_f$, the laminarity of~$\remsetL_f \subseteq \EE$ implies that the hyperedges in~$\remsetFIN_f$ are pairwise disjoint.     

To see $f \notin \remsetFIN_f$, it suffices to observe the following. Since $f$ blocks the current $\capac$-matching~$M_f$, for each~${v \in V}$ that is oversaturated in~$M_f \cup \{f\}$ we know that $f$ cannot be the worst hyperedge incident to~$v$ in~$M_f(v) \cup \{f\}$.
\end{claimproof}

Consider now the iteration during which $f$ is picked on line~\ref{line:pickedge}. If the algorithm adds~$f$ to the current $\capac$-matching~$M_f$ on line~\ref{line:pickedge}, then $M'_f=M_f \cup \{f\} \setminus \remsetFIN_f$, otherwise $M'_f=M_f$.
Hence, the above claim implies
$|M'_f(v)| \geq |M_f(v)|$ for each $v \in V$,
since $v$ can only be incident to at most one hyperedge in~$\remsetFIN_f$.
Moreover, if $|M_f(v)|=\capac(v)$ for some~$v \in V$, then the worst hyperedge in~$M'_f(v)$ is weakly preferred by~$v$ to the worst hyperedge in~$M_f(v)$: 
this is obvious if $v \notin f$, and if $v \in f$, then $v$ prefers~$f$ to at least one hyperedge in~$M_f(v)$, so $v$'s worst hyperedge cannot be worse than before, as $f$ is the only new hyperedge in $M_f'(v)$ not in $M_f(v)$ already.
Hence, whenever some vertex becomes saturated with hyperedges better than some $e \in \EE$, then this remains true during the remainder of the algorithm.
We refer to this observation as the \emph{saturation-monotonicity} of the algorithm.

To prove the stability of~$M$, consider an arbitrary edge~$f \in \EE \setminus M$ and the iteration during which~$f$ is examined at line~\ref{line:pickedge}. 
If $f$ is not added to~$M_f$ on line~\ref{line:addedge}, then some vertex~$v$ incident to~$f$ must be saturated in~$M_f$ with hyperedges that are all preferred to~$f$ by~$v$. By the saturation-monotonicity of the algorithm, this remains true for the returned $\capac$-matching~$M$, so $f$ cannot block~$M$. 
By contrast, if $f$ is added to~$M_f$ on line~\ref{line:addedge}, then $f$ must have been removed at some point on line~\ref{line:adjustM} during a later iteration as a consequence of adding some hyperedge~$e$ to the $\capac$-matching~$M_{e}$ on line~\ref{line:addedge} for which $f \in \remsetFIN_{e}$.
By our construction of~$\remsetFIN_{e}$, we know that $f$ is the worst hyperedge for some vertex $u \in V$ in~$M_e(u)$, so after removing $\remsetFIN_{e}$ from the $\capac$-matching~$M_e \cup \{e\}$, vertex~$u$ is saturated with hyperedges that are all preferred by~$u$ to~$f$. Again, by the saturation-monotonicity of the algorithm, this remains true for~$u$ during the algorithm, so $f$ cannot block~$M$. 
Thus, $M$ is indeed a stable $\capac$-matching for~$\HH$.

In every iteration, a hyperedge of~$\EE \setminus \CC$ gets added to~$\CC$, so there are $|\EE|$ iterations. As each iteration  can be performed in linear time, the running time is polynomial. 
\end{proof}

We complement Theorem~\ref{thm:laminar-shm-poly} by a hardness result for a very restricted weight function.
\begin{restatable}[\linkproof{sec:app-proof-laminar-weighted-nphard}]{theorem}{thmrcpconstrainednphard}
\label{thm:rc1p-constrained-nphard}    
\label{thm:laminar-weighted-nphard}    
The \maxwlamshm\ problem is $\NP$-complete  on laminar hypergraphs, even if all edges have weight~$0$ with the exception of a single edge of weight~$1$. 
\end{restatable}

\subsection{Subpath Hypergraphs}
\label{sec:rc1p}
 Moving our focus to instances of \shbm\ where the underlying hypergraph~$\HH$ is a \rcp\ hypergraph, we present an algorithm for the \maxwrcshbm\ problem that runs in polynomial time if 
the maximum size of a hyperedge in~$\HH$ is bounded by a constant.

\newcommand{\maxW}{\omega}
\begin{theorem}
\label{thm:RC1P-fixed-length}
 \maxwrcshbm\
can be solved in $O((\Delta (\capac_{\max}+1))^{\ell_{\max}} |\EE| \log \maxW)$ time for an instance ($\HH=(V,\EE),\capac,\weight,(\succ_v)_{v\in V})$, 
where $\capac_{\max}$ is the maximum capacity of a vertex, $\ell_{\max}$ is the maximum size of an edge in~$\EE$, $\Delta$ is the maximum degree of a vertex, and $\maxW$ is the maximum absolute value of $\weight (e)$ over~$e\in \EE$.
\end{theorem}

A brute-force approach to solve the \maxwrcshbm\ problem would be to consider the edges one-by-one and enumerate all possible $\capac$-matchings that do not violate the capacities and that may turn out to be stable in the end. 
However, such a brute-force method would not yield a polynomial-time algorithm even if $\ell_{\max}$ is constant.
Hence, we apply a more sophisticated approach  using dynamic programming: 
we process edges one-by-one, but instead of recording all possible partial solutions, 
we only maintain a smaller set of $\capac$-matchings that, in a certain sense, represents all possible strategies for obtaining a stable $\capac$-matching. 
We start by introducing the necessary concepts.

\medskip
\noindent
{\bf Strategies.}
A \emph{strategy} over a given set $W \subseteq V$ of vertices is a function $\sigma:W \rightarrow \HE \cup \{\nomatch\}$
such that $\sigma(w)$ is incident to~$w$ for each $w \in W$ unless $\sigma(w)=\nomatch$.
The interpretation of a strategy is that $\sigma(w)$ is the planned \emph{threshold} of~$w$, that is, 
the worst edge that we plan to put into the $\capac$-matching among those incident to~$w$; the symbol $\nomatch$ corresponds to being unsaturated.
We extend the notation~$\succ_w$ such that $e \succ_w \nomatch$ holds for each $e$ incident to~$w$.

\begin{Def}
A $\capac$-matching~$M$ \emph{realizes} a strategy~$\sigma$ over~$W$ w.r.t.\ a subset $\HE' \subseteq \HE$ of the edges, 
if 
\begin{itemize}
\item[(i)] $M$ is \emph{compatible} with~$\sigma$, meaning that for each $w \in W$: 
\begin{itemize}
    \item if $\sigma(w) \in \EE$, then $e \succeq_w \sigma(w)$ holds for every $e \in M(w)$, 
and
    \item if $\sigma(w)=\nomatch$, then $|M(w)|<c(w)$;
\end{itemize}
\item[(ii)] $\sigma(w) \in M(w)$ for each $w \in W$ with $\sigma(w) \in \HE'$,
and 
\item[(iii)] every edge in $\HE' \setminus M$ that blocks~$M$ is planned to be dominated by~$\sigma$ where we say that
$\sigma$ \emph{plans to dominate} an edge $e$ if there exists some $w \in W$ such that $\sigma(w) \succ_w e$.
\end{itemize}
\end{Def}
We call a strategy \emph{realizable} for an edge set~$\HE' \subseteq \HE$ if some $\capac$-matching realizes it w.r.t.~$\HE'$.

Observe that if a strategy~$\sigma$ plans to dominate an edge~$e$ at some vertex~$w$ 
and $M$ is a $\capac$-matching compatible with~$\sigma$ that also satisfies $|M(w)|=\capac(w)$,  
then $e$ cannot block~$M$, because $e$ is \emph{dominated} by~$M$ at~$w$, meaning that $M$ contains $\capac(w)$ edges incident to~$w$, 
and $w$ prefers each of them to~$e$.

Given a strategy~$\sigma$ over~$W$ and a $\capac$-matching compatible with~$\sigma$, 
we say that $M$ is \emph{complete} on a vertex set~$W' \subseteq W$ for~$\sigma$ 
if $|M(v)|=\capac(v)$ 
for each $v \in W'$ such that $\sigma(v) \neq \nomatch$.

\medskip
\noindent
{\bf The algorithm.}
Let $(\HH,\capac,\weight,(\succ_v)_{v \in V})$ be our instance of \maxwrcshbm\ with a \rcp\  hypergraph $\HH=(V,\EE)$, and let
$v_1,\dots,v_n$ be the ordering  of~$V$ witnessing the \rcp\ property of~$\HH$. That is, for each hyperedge $e\in \EE$ there exist indices $i\le j$ such that $e= \{ v_i,v_{i+1}, \dots, v_j\}$.  

For a hyperedge~$e \in \EE$, we write $l(e)=\max\{i:v_i \in e\}$ and we say that $v_{l(e)}$ is the \emph{endpoint} of~$e$.
Let $e_1,\dots, e_m$ be an ordering of $\HE$ where $l(e_i)<l(e_j)$ implies $i<j$, that is, we order the edges increasingly according to their endpoints.
Let $\HE_i=\{e_1,\dots,e_i\}$, and let $L_i$ contain the last $\ell_{\max}$ vertices contained in $\bigcup \HE_i$, that is, 
$L_i=\{v_j: \max\{ 0,l(e_i)-\ell_{\max} \} < j \leq l(e_i)\}$.

Our algorithm considers each set $\HE_i=\{e_1,\dots,e_i\}$, and
computes a set $S_i$ of pairs $(\sigma,M)$ where $\sigma$ is a strategy over~$L_i$ and $M$ is a $\capac$-matching in $(V,\HE_i)$ that realizes~$\sigma$.
We will ensure that $S_i$ is \emph{representative} for~$\HE_i$ in the following sense: 
\begin{Def}
Set $S_i$ is \emph{$\weight$-representative} for~$\HE_i$ if the following holds:
whenever there is a strategy~$\sigma$ and a $\capac$-matching~$M$ realizing~$\sigma$ with respect to~$\HE_i$, 
then there is a $\capac$-matching~$M'$ equivalent to~$M$ on~$L_i$ for which $(\sigma,M') \in S_i$ and $\weight(M') \geq \weight(M)$.
Here, $M$ is \emph{equivalent} to $M'$ on $L_i$, 
if $|M(v)|=|M'(v)|$ for each $v \in L_i$; this is denoted by $M \sim_{L_i} M'$.
\\
The operation of 
\emph{updating} $S_i$ with $(\sigma,M)$ means that 
we put $(\sigma,M)$ into~$S_i$
unless 
there is some pair $(\sigma,\hat{M})$ already in $S_i$ for which $\hat{M}$ is equivalent on~$L_i$ to $M$ and $\weight(\hat{M})\geq \weight(M)$.
\end{Def}

We initialize $S_0$ to contain only the pair $(\sigma_\emptyset,\emptyset)$, 
where $\sigma_\emptyset$ is the function with empty domain. 
We then iterate over $i=1$ to~$m$, and build the set~$S_i$ as follows;
for a formal description, see Algorithm~\ref{alg:RC1P}.

\begin{description}
\item[Adjusting the domain from~$L_{i-1}$ to~$L_i$:] 
For each $(\sigma,M) \in S_{i-1}$, check whether $M$ is complete for~$\sigma$ on~$L_{i-1} \setminus L_i$. 
If not, then we will not be able to make~$\sigma$ complete by adding edges from~$\HE \setminus \HE_{i-1}$, so we ignore the pair~$(\sigma,M)$.
Otherwise, we change the domain of~$\sigma$ from~$L_{i-1}$ to~$L_i$ (forgetting all values assigned to vertices of $L_{i-1} \setminus L_i$) as follows:
We set $\sigma(v)$ for each $v \in L_i \setminus L_{i-1}$ to all possible values in~$\{ e \in \HE:v \in e\} \cup \{\nomatch\}$, 
and for each strategy $\sigma'$ over $L_i$ obtained this way, we put $(\sigma',M)$ into a set~$T_i$.
\item[Computing $S_i$ based on $e_i$:]
For each $(\sigma,M) \in T_i$ we check whether $M$ still realizes $\sigma$ with respect to~$\HE_i$ by checking if (i) %
$\sigma$ plans to dominate~$e_i$, and 
(ii) for each $v \in e_i$, the threshold $\sigma(v)$ does not coincide with $e_i$.
If these conditions hold, then $M$ realizes~$\sigma$ with respect to~$\HE_i$. Hence, we update~$S_i$ with $(\sigma,M)$.
Next, we check whether we can add~$e_i$ to~$M$: if adding $e_i$ does not violate the capacities and yields a $\capac$-matching compatible with~$\sigma$, 
then we update $S_i$ with $(\sigma,M \cup\{e_i\})$.
\end{description}

The algorithm finishes by checking whether $S_m$ contains a pair~$(\sigma,M)$ such that~$M$  is complete for~$\sigma$ on~$L_m$, and if so, outputs among all such $\capac$-matchings~$M$ one that maximizes $\weight(M)$.

\begin{algorithm}
\caption{\maxwrcshbm$(\HH=(V,\EE),\capac,\weight,(\succ_v)_{v \in V})$ where $\HH$ is a subpath hypergraph}
\label{alg:RC1P}
\begin{algorithmic}[1]
\State Set $L_0:=\emptyset$ and $\sigma_\emptyset$ as the empty function over~$\emptyset$.
\State Set $S_0:=\{(\sigma_\emptyset,\emptyset)\}$.
\For {$i=1$ {\bf to} $m$} 
	\ForAll{pair~$(\sigma,M) \in S_{i-1}$}  		\label{alg2:lineA}
		\If{$M$ is not complete for $\sigma$ on $L_{i-1} \setminus L_{i}$} 	\label{alg2:lineA+1}
			\State {\bf break;} \Comment{Continue with next pair in $S_{i-1}$.}
		\EndIf
		\ForAll{$\sigma'': L_i \setminus L_{i-1} \rightarrow \HE \cup \{\nomatch\}$}
            \If{for each $v \in L_i \setminus L_{i-1}$, either $\sigma''(v) \ni v$ or $\sigma''(v)=\nomatch$} 
			 \State Set $\sigma'$ as the union $\sigma|_{L_{i-1} \cap L_i}$ and $\sigma''$, with domain~$L_i$.
			\State $T_i \leftarrow (\sigma',M)$		\label{alg2:lineB}
            \EndIf
		\EndFor
	\EndFor
	\State $S_i:=\emptyset$
	\ForAll{pair~$(\sigma,M) \in T_i$} 
		\If{$\sigma$ plans to dominate $e_i$ and $e_i \neq \sigma(v)$ for each~$v \in e_i$} 			\label{alg2:lineX-1}
			\If{$\not\exists (\sigma,\hat{M}) \in S_i: M \sim_{L_i} \hat{M} \text{ and } \weight(\hat{M})\geq \weight(M)$} $S_i \leftarrow (\sigma,M)$		\label{alg2:lineX}
			\EndIf
		\EndIf							\label{alg2:lineX+2}
		\If{$e_i \succeq_v \sigma(v)$ for each $v \in e_i$ and $\forall v \in e_i:|M(v)|<\capac(v)$} \label{alg2:lineY-2}
			\State $M':=M \cup \{e_i\}$
			\If{$\not\exists (\sigma,\hat{M}) \in S_i: M' \sim_{L_i} \hat{M} \text{ and } \weight(\hat{M})\geq \weight(M')$} $S_i \leftarrow (\sigma,M')$		\label{alg2:lineY}
			\EndIf
		\EndIf							\label{alg2:lineY+2}
	\EndFor
\EndFor
\State Set $\weight^\star:=-\infty$ and $M^\star:=\emptyset$.
\ForAll{pair~$(\sigma,M) \in S_m$} 
	\If{$M$ is complete for~$\sigma$ on~$L_m$ and $\weight(M)>\weight^\star$}			\label{alg2:line-Z} 
    \State 
        Set $\weight^\star:=\weight(M)$ and $M^\star:=M$.\label{alg2:line-findmax} 
    \EndIf
\EndFor
\State {\bf return} $M^\star$.			\label{alg2:line-ret}	
\end{algorithmic}
\end{algorithm}

\begin{restatable}[\linkproof{sec:app-proof-subpath_real}]{Lemma}{subpathlemmareal}
\label{lem:subpath_real}
If Algorithm~\ref{alg:RC1P} puts a pair~$(\sigma,M)$ into~$S_i$ for some $i \in [m]$, then $M$ realizes~$\sigma$ with respect to~$\HE_i$. 
\end{restatable}

\begin{proofsketch}
We use induction on~$i$. The case $i=0$ is trivial, so assume that the lemma holds for $i-1$.

First, suppose that $(\sigma,M)$ was added to~$S_i$ on line~\ref{alg2:lineX}, so $M \subseteq \HE_{i-1}$.
Let $(\sigma^-,M) \in S_{i-1}$ be the pair selected on line~\ref{alg2:lineA}
in the cycle during which $(\sigma,M)$ was put into~$T_i$ on line~\ref{alg2:lineB}. 
By induction, %
\begin{itemize}
    \item[(a)] $M$ is compatible with~$\sigma^-$,
    \item[(b)] $\sigma^-(v) \in M(v)$ for every $v \in L_{i-1}$ where $\sigma^-(v) \in \HE_{i-1}$, and 
    \item[(c)] every edge in $\EE_{i-1} \setminus M$ that blocks~$M$ is planned to be dominated by~$\sigma^-$.
\end{itemize}

First, observe that 
$M$ is compatible with~$\sigma$: the compatibility conditions for each $v \in L_{i-1} \cap L_i$ follow immediately from $\sigma^-(v)=\sigma(v)$ due to~(a), while the compatibility condition for each $v \in L_i \setminus L_{i-1}$ 
holds irrespective of the value of~$\sigma(v)$, because edges of~$M \subseteq \HE_{i-1}$ have no vertices in~$L_i \setminus L_{i-1}$. 

Second, by our assumption that 
Algorithm~\ref{alg:RC1P} added $(\sigma,M)$ to~$S_i$ on line~\ref{alg2:lineX}, we know that $\sigma(v) \neq e_i$ for each $v$ incident to~$e_i$. 
Hence, if $\sigma(v) \in \EE_i$ for some $v \in L_i$, then $\sigma(v) \in \EE_{i-1}$, and consequently, $v \in L_i \cap L_{i-1}$ (because vertices of $L_i \setminus L_{i-1}$ are not contained in any edge of~$\EE_{i-1}$). 
Therefore, (b) implies
$\sigma(v)=\sigma^-(v) \in M(v)$. 
Third, it remains to show that every edge $e \in \HE_i \setminus M$ that blocks $M$ is planned to be dominated by~$\sigma$.
First note that Algorithm~\ref{alg:RC1P} checks that this holds for~$e_i$ explicitly on line~\ref{alg2:lineX-1}, so we may assume $e \neq e_i$.
By (c), we know that every edge in~$\HE_{i-1} \setminus M$ that blocks~$M$ is planned to be dominated by~$\sigma^-$, 
so there is some $v \in L_{i-1}$ such that $\sigma^-(v) \succ_v e$. On the one hand, if $v \in L_i \cap L_{i-1}$, then $e$ is planned to be dominated by $\sigma$ too, because $\sigma(v)=\sigma^-(v)$.
On the other hand, if $v \in L_{i-1} \setminus L_{i}$, then Algorithm~\ref{alg:RC1P} must have confirmed on line~\ref{alg2:lineB}
that $M$ is complete on~$v$ (otherwise it would not have put $(\sigma^-,M)$ into~$T_i$),  which implies that $M$ dominates~$e$ at~$v$. Thus, $e$ does not block~$M$, proving 
that every edge in $\EE_i \setminus M$ that blocks~$M$ is planned to be dominated by~$\sigma$.
This finishes the proof that $M$ realizes~$\sigma$ with respect to~$\HE_i$.

The case when  $(\sigma,M)$ was put into~$S_i$ on line~\ref{alg2:lineY} is similar; see Appendix~\ref{sec:app-rc1p-proofs} for the remainder.

\end{proofsketch}

\begin{Lemma}
\label{lem:RC1P-represent}
Set~$S_i$ computed by Algorithm~\ref{alg:RC1P} is $\weight$-representative for~$\HE_i$ for  each $i\in [m]$. 
\end{Lemma}

\begin{proof}
We prove the lemma by induction on~$i$. Clearly, the lemma trivially holds for $i=0$, so 
let us assume that the statement holds for $i-1$.

To prove that $S_i$ remains $\weight$-representative for~$\HE_i$, it 
suffices to show the following: if $(\sigma,M)$ is such that $M \subseteq \HE_i$ realizes $\sigma$ with respect to~$\HE_i$, 
then the algorithm 
updates $S_i$ with \emph{some} pair $(\sigma,M')$ 
either at line~\ref{alg2:lineX} or at line~\ref{alg2:lineY} for which $M' \sim_{L_i} M$
(note that %
$\sim_{L_i}$ is transitive) and $\weight(M') \geq \weight(M)$. 

Given $\sigma$ and $M$, let us define the strategy $\sigma^-$ over~$L_{i-1}$ as follows:
\[
\sigma^-(v) = \left\{ 
\begin{array}{lp{8.5cm}} 
\sigma(v), & \textrm{if $v \in L_{i-1} \cap L_i$}; \\
\nomatch, & \textrm{if $v \in L_{i-1} \setminus L_i$ and $|M(v)|<\capac(v)$}; \\
\worst(M,v), & \textrm{if $v \in L_{i-1} \setminus L_i$ and $|M(v)|=\capac(v)$}
\end{array} \right.
\] 
where $\worst(M,v)$ is the worst edge in~$M(v)$ according to~$\succ_v$. Note that $M$ is complete on vertices of~$L_{i-1} \setminus L_i$ for~$\sigma^-$.
We distinguish between two cases.

{\bf Case A:} $e_i \notin M$. 
We claim that $M$ realizes $\sigma^-$ w.r.t.~$\HE_{i-1}$.
First, since $M$ realizes $\sigma$ w.r.t.~$\HE_i$, we know that 
(i) $M$ satisfies the compatibility conditions for $\sigma^-$ on vertices of $L_i \cap L_{i-1}$, because $\sigma$ coincides with $\sigma^-$ on these vertices, and 
(ii) $\sigma^-(v)=\sigma(v) \in M(v)$ for each $v \in L_i \cap L_{i-1}$ that satisfies $\sigma(v) \in \HE_i$.
Considering some $v \in L_{i-1} \setminus L_i$, by the definition of~$\sigma^-$, 
we get that $M$ is compatible with~$\sigma^-$, 
and also that $\sigma(v) =\worst(M,v) \in M(v)$ for every $v \in L_{i-1} \setminus L_i$ with $\sigma(v) \neq \nomatch$.
Hence, to see the claim, we only need to show that $\sigma^-$ plans to dominate every edge $e \in \HE_{i-1} \setminus M \subseteq \HE_i \setminus M $ that blocks~$M$.
Clearly, $e$ is planned to be dominated by~$\sigma$ at some $v \in L_i$
(since $M$ realizes $\sigma$ w.r.t.~$\HE_i$). However, as $e \in \HE_{i-1}$, it follows that $v \in L_{i-1} \cap L_i$, 
and thus $\sigma^-$ indeed plans to dominate~$e$ at~$v$ by $\sigma^-(v)=\sigma(v)$.
This proves the claim that $M$ realizes $\sigma^-$ with respect to~$\HE_{i-1}$.
Therefore, by induction there exists a $\capac$-matching~$M'$ with $\weight(M') \geq \weight(M)$ such that $(\sigma^-,M') \in S_{i-1}$
and $M' \sim_{L_{i-1}} M$. 

Recall that $M$ is complete on $L_{i-1} \setminus L_i$ for~$\sigma^-$, 
and thus $M' \sim_{L_{i-1}} M$ implies that $M'$ is also complete on~$L_{i-1} \setminus L_i$ for~$\sigma^-$.
Hence, $(\sigma^-,M')$ passes the test on line~\ref{alg2:lineA+1}. 
As Algorithm~\ref{alg:RC1P} tries all possible ways to extend~$\sigma^-$, it follows that $(\sigma,M') \in T_i$.
Now, since $M$ realizes $\sigma$ w.r.t.~$\HE_i$ but $e_i \notin M$,
we get that $e_i \neq \sigma(v)$ for any $v \in e_i$. Additionally, $\sigma$ plans to dominate $e_i$: if $e_i$ blocks $M$, then this is because $M$ realizes $\sigma$ w.r.t.\ $\EE_i$ (by our assumptions),
and if $e_i$ does not block~$M$, then it is dominated at some vertex~$v \in L_i$, which in turn implies that $v$ is saturated by~$M$ and $\worst(M,v)=\sigma(v) \succ_v e_i$, that is, $\sigma$ indeed plans to dominate~$e_i$.
Therefore, Algorithm~\ref{alg:RC1P} 
updates~$S_i$ with $(\sigma,M')$
at line~\ref{alg2:lineX}.%

{\bf Case B:} $e_i \in M$. 
 Consider the $\capac$-matching $M^-=M \setminus \{e_i\}$;
we claim that $M^-$ realizes $\sigma^-$ w.r.t.~$\HE_{i-1}$.
Again, since $M$ realizes $\sigma$ w.r.t.~$\HE_i$, we know that 
(i)~$M^- \subseteq M$ satisfies the compatibility conditions for~$\sigma^-$ on vertices of~$L_i \cap L_{i-1}$, because $\sigma$ coincides with $\sigma^-$ on these vertices, and 
(ii)~$\sigma^-(v)=\sigma(v) \in M^-(v)$ for each $v \in L_i \cap L_{i-1}$ that satisfies $\sigma(v) \in \HE_{i-1}=\EE_i \setminus \{e_i\}$.
By the definition of~$\sigma^-$ on $L_{i-1} \setminus L_i$, it follows that $M^-$ is compatible with~$\sigma^-$, 
and  $\sigma(v) \in M^-(v)$ for each $v \in L_{i-1}$ with $\sigma(v) \in \HE_{i-1}$. 
Hence, to see the claim, we only need to show that $\sigma^-$ plans to dominate every edge $e \in \HE_{i-1} \setminus M^-=\HE \setminus M$ that blocks~$M^-$. 
If $e$ blocks $M$ as well, then $\sigma$ plans to dominate it at some $v \in L_i$
(since $M$ realizes~$\sigma$ w.r.t.~$\HE_i$). Since $e \in \HE_{i-1}$, it follows that $v \in L_{i-1} \cap L_i$, 
and thus $\sigma^-$ plans to dominate~$e$ as well. 
If, by contrast, $e$ does not block~$M$, then there must exist some $v \in L_i \cap e_i$ 
where $M$ dominates $e$. Then, $|M(v)|=\capac(v)$ and $\worst(M,v)=\sigma(v) \succ_v e$ (since $M$ realizes $\sigma$ w.r.t.~$\HE_i$), 
and so $\sigma^-(v) \succ_v e$ by the definition of $\sigma^-$.
Hence, $\sigma^-$ plans to dominate~$e$ at~$v$.
This proves the claim that $M^-$ realizes $\sigma^-$ with respect to~$\HE_{i-1}$. 
Therefore, by induction there exists a $\capac$-matching~$M'$ with $\weight(M') \geq \weight(M^-)$ such that $(\sigma^-,M') \in S_{i-1}$
and $M' \sim_{L_{i-1}} M^-$. 

Recall that $M^-$ is complete on $L_{i-1} \setminus L_i$ for~$\sigma^-$, 
and thus $M' \sim_{L_{i-1}} M^-$ implies that $M'$ is also complete on~$L_{i-1} \setminus L_i$ for~$\sigma^-$.
Therefore, $(\sigma^-,M')$ passes the test on line~\ref{alg2:lineA+1}. 
As Algorithm~\ref{alg:RC1P} tries all possible ways to extend~$\sigma^-$, it follows that $(\sigma,M') \in T_i$.
Since $e_i \in M$ and $M$ is compatible with~$\sigma$, we know that $e_i \succeq_v \sigma(v)$  for each $v \in L_i$.
By $M' \sim_{L_{i-1}} M^-$ and since $M^- \cup \{e_i\}=M$ is a $\capac$-matching, we know $|M'(v)|<\capac(v)$ for each $v \in e_i \cap L_{i-1}$;
since $M' \subseteq \HE_{i-1}$, it is also clear that $|M'(v)|=0<\capac(v)$ for each $v \in L_i \setminus L_{i-1}$. 
Therefore, the pair $(\sigma,M')$ satisfies the conditions on line~\ref{alg2:lineY-2} and so 
Algorithm~\ref{alg:RC1P}
updates $S_i$ with $(\sigma,M' \cup \{e_i\})$
at line~\ref{alg2:lineY}. By $M' \cup \{e_i\} \sim_{L_i} M^- \cup \{e_i\} =M$ and $\weight(M' \cup \{e_i\})\ge \weight(M^-)+\weight(e_i)= \weight(M)$, this finishes the proof of the lemma.%
\end{proof}

\begin{proof}[Proof of Theorem~\ref{thm:RC1P-fixed-length}]
We start by proving that Algorithm~\ref{alg:RC1P} solves {\maxwrcshbm} correctly on \rcp\ hypergraphs.

First, we show that the output is stable. Observe that if  Algorithm~\ref{alg:RC1P} returns a $\capac$-matching~$M$ on line~\ref{alg2:line-ret}, 
then $M$ realizes $\sigma$ on~$\HE$, and thus $\sigma$ plans to dominate every edge blocking~$M$. 
However, as $M$ is complete for $\sigma$ (due to the condition on line~\ref{alg2:line-Z}), $M$ in fact dominates every edge that $\sigma$ plans to dominate. 
Thus, no edge may block~$M$.

Assume now that $M$ is a maximum-weight stable $\capac$-matching in our instance. 
Let $\sigma$ be the strategy over $L_m$ defined as 
\begin{equation*} \label{eq:sigma-for-M} 
\sigma(v) = \left\{ 
\begin{array}{ll} 
\nomatch, & \textrm{if $|M(v)|<\capac(v)$}; \\
\worst(M,v), & \textrm{if $|M(v)|=\capac(v)$}.
\end{array} \right.
\end{equation*}

It is clear that $M$ realizes~$\sigma$ over $\HE$. 
Thus, by Lemma~\ref{lem:RC1P-represent} there exists a $\capac$-matching~$M'$ satisfying $\weight(M') \geq \weight(M)$ for which $(\sigma,M') \in S_m$ and $M \sim_{L_m} M'$. 
Moreover, as $M$ is complete on~$L_m$ for~$\sigma$ and $M \sim_{L_m} M'$, we know that $M'$ is complete on~$L_m$ for $\sigma$ 
as well.
Hence, Algorithm~\ref{alg:RC1P} will set $\weight^\star$ to at least~$\weight(M') \geq \weight(M)$ 
on line~\ref{alg2:line-findmax} at latest when it examines the pair $(\sigma,M')$, and will thus 
return a stable $\capac$-matching of weight at least $\weight(M)$ on line~\ref{alg2:line-ret}, which must be a maximum-weight stable $\capac$-matching. 

To bound the running time of Algorithm~\ref{alg:RC1P}, 
note that the number of possible strategies for a given set~$L_i$ is at most $\Delta^{|L_i|} \leq \Delta^{\ell_{\max}}$
where $\Delta$ is the maximum number of edges incident to any vertex. 
Given a strategy~$\sigma$ over $L_i$, we also know $|\{M:(\sigma,M) \in S_i\}| \leq (\capac_{\max}+1)^{|L_i|} \leq (\capac_{\max}+1)^{\ell_{\max}}$
since this set does not contain two $\capac$-matchings that are equivalent on~$L_i$. This yields
\[|S_i| \leq (\Delta(\capac_{\max}+1))^{\ell_{\max}}.\]
There are $m$ iterations, and the computation of $S_i$ from~$S_{i-1}$ is proportional to~${|S_{i-1}|+|S_i|}$, 
as any line of the algorithm can be performed in $O(\ell_{\max} \log \maxW)$ time. 
Thus,
we obtain that the total running time is $O((\Delta(\capac_{\max}+1))^{\ell_{\max}}m\log \maxW)$, 
as promised.
\end{proof}

\begin{Rem}
    We remark that Algorithm~\ref{alg:RC1P} can be used to decide whether  a stable $\capac$-matching exists that contains  a given set of ``forced'' edges (or is disjoint from a given set ``forbidden'' edges) by setting an appropriate weight function.
    Furthermore, it can also be used to 
     decide the existence of a stable $\capac$-matching that saturates (or does not saturate) a given vertex~$v$. 
     Indeed, to decide whether some stable $\capac$-matching saturates~$v$, it suffices to add a newly created hyperedge~$e_v$ incident only to~$v$, extend the preferences of~$v$ such that $e_v$ becomes the least-favorite edge according to~$\succ_v$, and decide whether this modified instance admits a stable $\capac$-matching that does not contain~$e_v$.
     Deciding whether some stable $\capac$-matching leaves $v$ unsaturated can be done analogously, using the same construction but searching for a stable $\capac$-matching that contains~$e_v$.
\end{Rem}

Theorem~\ref{thm:RC1P-Whardness} shows that the running time of the algorithm in Theorem~\ref{thm:RC1P-fixed-length} cannot be improved to an FPT algorithm with parameter~$\ell_{\max}$, not even if we additionally parameterize the problem with the maximum capacity and severely restrict the weight function.

\begin{restatable}[\linkproof{sec:app-proof-RC1P-Whardness}]{theorem}{thmrcpwhardness}
\label{thm:RC1P-Whardness}
The \maxwrcshbm{} problem is $\mathsf{W}[1]$-hard when parameterized by $\capac_{\max} +\ell_{\max}$, where $\capac_{\max}$ is the maximum capacity of a vertex and $\ell_{\max}$ is the maximum size of an edge, 
even if all edges have weight~$0$ with the exception of a single edge of weight~$1$.
\end{restatable}

\section{\texorpdfstring{\shbm\ for Subtree Hypergraphs}{SHbM for subtree hypergraphs}}\label{sec:subtree}

We now turn our focus to the class of subtree hypergraphs. As these are not necessarily unimodular, so an instance of \subtrshbm\ may not admit a stable $\capac$-matching. However, as subtree hypergraphs are still normal, \subtrshm\ is guaranteed to be solvable.

We start by presenting a polynomial-time algorithm that solves \subtrshm.
Before explaining our algorithm, let us introduce some additional notation. Let $\HH=(V,\EE)$ be a subtree hypergraph and let $T=(V,F)$ be the underlying tree, that is, for each hyperedge~${e \in \EE}$, the set of vertices contained in~$e$  induces a subtree of $T$. 
Let $r$ be a fixed vertex of $V$ and consider the rooted tree $T$ with root $r$. For each vertex $v$, let $\dist_T(r,v)$ be the \emph{distance} of~$r$ and~$v$ in~$T$, that is the number of edges of the (unique) path of~$T$ between~$r$ and~$v$.
For each hyperedge $e \in \EE$, let $\mytop_{T,r}(e)$ be the (unique) vertex of $e$ that is closest to~$r$ in~$T$, and define $\dist_{T,r}(e):= \dist_T(r,\mytop_{T,r}(e))$.

Given an instance $(\HH,(\succ_v)_{v \in V})$ of \subtrshm, where $\HH=(V,\EE)$ is a subtree hypergraph with underlying tree~$T$, algorithm Subtree-\shm\ does the following (see Algorithm~\ref{alg:subtree} for a pseudocode). First, it finds a hyperedge~$f$ such that $\mytop_{T,r}(f)$ is furthest from~$r$ in~$T$; for simplicity, let us write~$r_f$ for $\mytop_{T,r}(f)$. Then it deletes~$f$ and every hyperedge incident to~$r_f$ that are less preferred by~$r_f$ than~$f$, and recursively calls itself on the remaining hypergraph. Finally, it builds a matching by adding $f$ to the output~$M$ of the recursive call if $r_f$ is not covered by~$M$.
\begin{algorithm}[t]
\caption{Subtree-\shm%
$(\HH,(\succ_v)_{v \in V})$ where $\HH=(V,\EE)$ is a subtree hypergraph with underlying tree~$T$ rooted at~$r \in V$}

\label{alg:subtree}
\begin{algorithmic}[1]
\If{$\EE=\emptyset$} \textbf{return} $M=\emptyset$.
\EndIf

\State $f:=\textrm{argmax}_{e\in \EE}\dist_{T,r}(e)$
\State $r_f:=\mytop_{T,r}(f)$
\State $\remsetL_f := \{ f'\in \EE\mid r_f \in f' ,  f\succ_{r_f}f' \}$
\State $\EE':=\EE \setminus (\remsetL_f\cup \{ f\})$
\State $M=$ Subtree-\shm%
$((V,\EE'),(\succ_v)_{v \in V})$

\If{$r_f$ is not covered by $M$} $M:= M\cup \{f\}$.
\EndIf
\State \textbf{return} $M$.
 
\end{algorithmic}
\end{algorithm}

\begin{theorem}
Algorithm \ref{alg:subtree} solves \subtrshm\ in polynomial time.
\label{thm:subtree-poly}
\end{theorem}
\begin{proof}

The algorithm runs in polynomial time, because Subtree-\shm\ is called at most $|\EE|$ times recursively, and in each call, it only needs to compute the farthest hyperedge~$f$ and delete those hyperedges incident to~$r_f=\mytop_{T,r}(f)$ that are worse than~$f$ for~$r_f$.

We show that the output $M$ is a stable hypergraph matching. We use induction on $|\EE|$. If $\EE=\emptyset$, then the output $M=\emptyset$ of  Algorithm~\ref{alg:subtree} is clearly a stable hypergraph matching.

Hence, assume that $\EE\ne \emptyset$. Let $f$ be the hyperedge chosen by the algorithm in the first iteration. By induction, %
$M$ is a stable hypergraph matching in the subtree hypergraph $(V,\EE\setminus (\remsetL_f\cup \{f\} ))$. 

First, suppose that $M$ saturates $r_f$, so Algorithm \ref{alg:subtree} outputs $M$. Then, $M$ is clearly a matching. Also, no hyperedge of $\EE\setminus (\remsetL_f\cup \{f\} )$ blocks $M$ by induction. Finally, no hyperedge of $\remsetL_f\cup \{ f\}$ blocks~$M$ either, because $r_f$ is  covered with some hyperedge $e \in M$ such that $e\succ_{r_f} f\succ_{r_f}f'$ for every $f'\in \remsetL_f$ by the choice of~$\remsetL_f$. 

Suppose now that $M$ does not saturate $r_f$, and the output is therefore ${M\cup \{ f\}}$. The fact that $\HH$ is a subtree hypergraph and that $f$ is chosen so that it  maximizes $\dist_{T,r}(f)$ implies that any hyperedge that intersects~$f$ must also contain~$r_f$. Hence, as $M$ does not cover~$r_f$, $M\cup \{ f\}$ is a matching. By induction, no hyperedge of $\EE\setminus (\remsetL_f\cup \{f\} )$ blocks $M$, and thus $M\cup \{ f\}$. Finally, no hyperedge of $\remsetL_f$ blocks $M\cup \{ f\}$, because $f\succ_{r_f}f'$ for every $f'\in \remsetL_f$ and $r_f$ is saturated by $M\cup \{ f\}$.

We conclude that the output is a stable matching. 
\end{proof}

A natural question to ask whether a stable $\capac$-matching can still be found efficiently in the \shbm\ setting. Sadly, the answer is no, even with very severe restrictions.

\begin{restatable}[\linkproof{sec:app-proof-thmfa}]{theorem}{thmfa}  
\label{thm:fa}
Deciding if there exists an stable hypergraph $\capac$-matching in an instance of \subtrshbm\ is $\NP$-hard, even if the underlying tree is a star, only one vertex has capacity greater than 1, and all hyperedge sizes are at most~4.
\end{restatable}

\begin{Rem}
Using the reduction presented in the proof of Theorem~\ref{thm:fa} and the fact that the problem of finding a stable \emph{fractional} hypergraph matching is $\PPAD$-complete even if each hyperedge has size~$3$~\cite{csaji2022complexity}, it is straightforward to verify that the problem of finding a stable fractional hypergraph matching in an instance of \subtrshbm\ satisfying the conditions of Theorem~\ref{thm:fa} is $\PPAD$-complete too.  
\end{Rem}

\section{Conclusions}
\label{sec:conclusion}
We studied stable hypergraph matchings and $\capac$-matchings, focusing on structured hypergraph families. We introduced the University Dual Admission problem and showed it can be regarded as a special case of \tushbm. 
We identified multiple classes of hypergraphs where stable ($\capac$-)matchings can be found efficiently, and provided contrasting proofs of intractability;  see Table~\ref{tab:results} for a summary.

Besides some open questions in Table~\ref{tab:results}, the main problem we leave open for further research is whether it is possible to find a stable matching in an arbitrary instance of \tushm\ in polynomial time. Our complexity results seem to suggest that the problem might be computationally hard (e.g., PPAD-hard) in general; however, for special cases there is a potential for  efficient combinatorial algorithms.
In particular, resolving the complexity of this problem for network hypergraphs would be of great interest. 
More generally, the development of new techniques to attain stable solutions for instances of \shm\ and \shbm\ is an important direction for future research that is motivated by various practical application such as the \textsc{Hospitals/Residents with Couples} problem~\cite[Chapter 5.3]{manlove2013algorithmics}.

\section*{Funding}
{All authors acknowledge the financial support by the Hungarian Academy of Sciences under its Momentum Programme, grant number LP2021-2.
Gergely Csáji and Péter Biró acknowledge the financial support by the Hungarian Scientific Research Fund via OTKA grant
number K143858.
Gergely Csáji is further supported by the Ministry of Culture and Innovation of Hungary from the National Research, Development
and Innovation fund, financed under the KDP-2023 funding scheme (grant number C2258525). 
Ildik\'o Schlotter is supported by the Hungarian Academy of Sciences
under its J\'anos Bolyai
Research Scholarship. 

\begin{thebibliography}{10}

\bibitem{AIM07}
David~J. Abraham, Robert~W. Irving, and David~F. Manlove.
\newblock Two algorithms for the {S}tudent--{P}roject {A}llocation problem.
\newblock {\em Journal of Discrete Algorithms}, 5(1):73--90, 2007.
\newblock \href {https://doi.org/10.1016/j.jda.2006.03.006} {\path{doi:10.1016/j.jda.2006.03.006}}.


\bibitem{AharoniFleiner03}
Ron Aharoni and Tamás Fleiner.
\newblock On a lemma of {S}carf.
\newblock {\em Journal of Combinatorial Theory, Series B}, 87(1):72--80, 2003.
\newblock \href {https://doi.org/10.1016/S0095-8956(02)00028-X} {\path{doi:10.1016/S0095-8956(02)00028-X}}.

\bibitem{aziz2012existence}
Haris Aziz and Florian Brandl.
\newblock Existence of stability in hedonic coalition formation games.
\newblock In {\em Proceedings of the 11th International Conference on Autonomous Agents and Multiagent Systems (AAMAS 2012)}, volume~2, pages 763--770, 2012.

\bibitem{banerjee2001core}
Suryapratim Banerjee, Hideo Konishi, and Tayfun S{\"o}nmez.
\newblock Core in a simple coalition formation game.
\newblock {\em Social Choice and Welfare}, 18(1):135--153, 2001.
\newblock \href {https://doi.org/10.1007/s003550000067} {\path{doi:10.1007/s003550000067}}.

\bibitem{berge-hypergraphs}
Claude Berge.
\newblock {\em Hypergraphs}.
\newblock North-Holland, 1989.


\bibitem{biro2016fractional}
P{\'e}ter Bir{\'o} and Tam{\'a}s Fleiner.
\newblock Fractional solutions for capacitated {NTU}-games, with applications to stable matchings.
\newblock {\em Discrete Optimization}, 22:241--254, 2016.
\newblock \href {https://doi.org/10.1016/j.disopt.2015.02.002} {\path{doi:10.1016/j.disopt.2015.02.002}}.

\bibitem{bretto-hypergraphs}
Alain Bretto.
\newblock {\em Hypergraph Theory: An Introduction}.
\newblock Mathematical Engineering. Springer Cham, 2013.

\bibitem{chandrasekaran2024scarf}
Karthekeyan Chandrasekaran, Yuri Faenza, Chengyue He, and Jay Sethuraman.
\newblock Scarf's algorithm on arborescence hypergraphs.
\newblock arXiv:2412.03397 [cs.DM], 2024.
\newblock \href {https://doi.org/10.48550/arXiv.2412.03397} {\path{doi:10.48550/arXiv.2412.03397}}.

\bibitem{csaji2022complexity}
Gergely Cs{\'a}ji.
\newblock On the complexity of stable hypergraph matching, stable multicommodity flow and related problems.
\newblock {\em Theoretical Computer Science}, 931:1--16, 2022.
\newblock \href {https://doi.org/10.1016/j.tcs.2022.07.025} {\path{doi:10.1016/j.tcs.2022.07.025}}.

\bibitem{faenza2023scarf}
Yuri Faenza, Chengyue He, and Jay Sethuraman.
\newblock Scarf's algorithm and stable marriages.
\newblock arXiv:2303.00791 [math.CO], 2023.
\newblock \href {https://doi.org/10.48550/arXiv.2303.00791} {\path{doi:10.48550/arXiv.2303.00791}}.


\bibitem{FHRV2009-multicolored-hardness}
Micheal~R. Fellows, Danny Hermelin, Frances Rosamond, and St\'ephane Vialette.
\newblock On the parameterized complexity of multiple-interval graph problems.
\newblock {\em Theoretical Computer Science},
410(1):53--61, 2009.
\newblock \href {https://doi.org/10.1016/j.tcs.2008.09.065} {\path{doi:10.1016/j.tcs.2008.09.065}}.


\bibitem{fleiner-ferkai-biro2019}
Rita Fleiner, Andr{\'a}s Ferkai, and P{\'e}ter Bir{\'o}.
\newblock College admission problem for university dual education.
\newblock In {\em Proceedings of the IEEE 17th World Symposium on Applied Machine Intelligence and Informatics (SAMI 2019)}, pages 31--36. IEEE, 2019.
\newblock \href {https://doi.org/10.1109/SAMI.2019.8782783} {\path{doi:10.1109/SAMI.2019.8782783}}.



\bibitem{Fleiner03-bmatching}
Tamás Fleiner.
\newblock On the stable $b$-matching polytope.
\newblock {\em Mathematical Social Sciences}, 46(2):149--158, 2003.
\newblock \href {https://doi.org/10.1016/S0165-4896(03)00074-X} {\path{doi:10.1016/S0165-4896(03)00074-X}}.

\bibitem{gale1962college}
David Gale and Lloyd~S. Shapley.
\newblock College admissions and the stability of marriage.
\newblock {\em The American Mathematical Monthly}, 69(1):9--15, 1962.
\newblock \href {https://doi.org/10.2307/2312726} {\path{doi:10.2307/2312726}}.

\bibitem{haeringer2018market}
Guillaume Haeringer.
\newblock {\em Market Design: Auctions and Matching}.
\newblock MIT Press, 2018.

\bibitem{kintali2013reducibility}
Shiva Kintali, Laura~J. Poplawski, Rajmohan Rajaraman, Ravi Sundaram, and Shang-Hua Teng.
\newblock Reducibility among fractional stability problems.
\newblock {\em SIAM Journal on Computing}, 42(6):2063--2113, 2013.
\newblock \href {https://doi.org/10.1137/120874655} {\path{doi:10.1137/120874655}}.

\bibitem{lovasz1972normal}
L{\'a}szl{\'o} Lov{\'a}sz.
\newblock Normal hypergraphs and the perfect graph conjecture.
\newblock {\em Discrete Mathematics}, 2(3):253--267, 1972.
\newblock \href {https://doi.org/10.1016/0012-365X(72)90006-4} {\path{doi:10.1016/0012-365X(72)90006-4}}.

\bibitem{manlove2013algorithmics}
David~F. Manlove.
\newblock {\em Algorithmics of Matching Under Preferences}, volume~2 of {\em Series on Theoretical Computer Science}.
\newblock World Scientific, 2013.
\newblock \href {https://doi.org/10.1142/8591} {\path{doi:10.1142/8591}}.

\bibitem{MIIMM02}
David~F. Manlove, Robert~W. Irving, Kazuo Iwama, Shuichi Miyazaki, and Yasufumi Morita.
\newblock Hard variants of stable marriage.
\newblock {\em Theoretical Computer Science}, 276(1--2):261--279, 2002.
\newblock \href {https://doi.org/10.1016/S0304-3975(01)00206-7} {\path{doi:10.1016/S0304-3975(01)00206-7}}.


\bibitem{NgHirschberg91}
Cheng Ng and Daniel~S. Hirschberg.
\newblock Three-dimensional stable matching problems.
\newblock {\em SIAM Journal on Discrete Mathematics}, 4(2):245--252, 1991.
\newblock \href {https://doi.org/10.1137/0404023} {\path{doi:10.1137/0404023}}.


\bibitem{roth1984}
Alvin~E. Roth.
\newblock The evolution of the labor market for medical interns and residents: A case study in game theory.
\newblock {\em Journal of Political Economy}, 92(6):991--1016, 1984.
\newblock \href {https://doi.org/10.1086/261272} {\path{doi:10.1086/261272}}.

\bibitem{roth-sotomayor1990}
Alvin~E. Roth and Marilda A.~O. Sotomayor.
\newblock {\em Two-sided Matching: a Study in Game-theoretic Modeling and Analysis}.
\newblock Econometric Society Monographs, Cambridge, 1990.

\bibitem{scarf1967core}
Herbert~E. Scarf.
\newblock The core of an {N} person game.
\newblock {\em Econometrica}, 35(1):50--69, 1967.
\newblock \href {https://doi.org/10.2307/1909383} {\path{doi:10.2307/1909383}}.

\bibitem{Schrijver}
Alexander Schrijver.
\newblock {\em Combinatorial Optimization}.
\newblock Springer-Verlag Berlin Heidelberg, 2003.

\bibitem{woeginger2013core}
Gerhard~J. Woeginger.
\newblock Core stability in hedonic coalition formation.
\newblock In {\em Proceedings of the 39th International Conference on Current Trends in Theory and Practice of Computer Science (SOFSEM 2013)}, volume 7741 of {\em Lecture Notes in Computer Science}, pages 33--50. Springer, 2013.
\newblock \href {https://doi.org/10.1007/978-3-642-35843-2_4} {\path{doi:10.1007/978-3-642-35843-2_4}}.

\end{thebibliography}

\clearpage
\begin{appendices}

\section{Motivation by Applications}
\label{sec:app-motivapplic}

     Instances of \shm\ and \shbm\ on special hypergraph classes arise when there is some underlying structure over participating agents. Such structure may be due to the physical positions of agents, organizational constraints, or some other feature of the model. Besides the real-life application of dual admission in Hungarian universities already mentioned~\cite{fleiner-ferkai-biro2019}, we briefly describe three additional examples:

\begin{description}
    \item[Budgeting in transportation networks:] Suppose that a set of cities aims to collaborate in order to renovate certain roads (or other infrastructure) running between the cities. Assume that there are several possible contracts to choose from, each  involving a subset of the cities (and the roads pertaining to these cities), with each city having preferences over the contracts it is involved in. Such a scenario can be formulated as an instance of \shm\ with cities as agents and contracts as hyperedges.
    In particular, if cities are positioned along a single road  or they are connected in a treelike fashion, 
    then the underlying hypergraph may be a subpath or a subtree hypergraph, respectively.
    \item[Project selection management:] Consider a firm whose organizational structure is a rooted tree~$T$. Suppose that there is a set of possible projects, each of them to be undertaken by different working groups, and an optimal subset of the projects needs to be selected. Naturally, the firm's employees may have different preferences over the projects, and employees may  have an upper bound on the number of projects they can participate in. 
    To incentivize employees, notions of stability may become important, and hence such a situation can be modeled as an instance of \shbm\ where agents are employees and hyperedges represent projects. 
    If each working group responsible for a project   forms a subtree of~$T$, %
    the underlying hypergraph is a subtree hypergraph. 
    In the special case when working groups always include all subordinates of any participant, %
    the underlying hypergraph is laminar.
    \item[Coalition formation in politics:] Often, political parties can be represented on a left–right axis based on their ideology. To gain power, parties may need to form coalitions, and it is reasonable to assume that possible coalitions %
    form subpaths along this axis. Then coalition formation can be thought of as an instance of \shm\ over a subtree hypergraph.
\end{description}

\section{Additional Material on Section~\ref{sec:prelim}}

\subsection{Graph-Theoretic Basics}
\label{sec:app-graphs}

Let us provide here some basic definitions and notation related to directed and undirected graphs.

An undirected (directed) graph consists of set~$V$ of vertices and a set~$E$ of edges (arcs, respectively). The edge \emph{connecting} two vertices~$u$ and~$v$ of an undirected graph is denoted by~$\{u,v\}$, while the arc \emph{leading from~$u$ to~$v$} in a directed graph is denoted by~$(u,v)$; in both cases, $u$ and $v$ are the called the endpoints of the edge or arc. 

Given an undirected or directed graph~$G$, a series~$e_1,e_2,\dots, e_k$ of edges or arcs in $G$ forms a \emph{cycle} if there are vertices $v_1,v_2,\dots, v_\ell$ in~$G$ such that for each $i \in [\ell]$, 
(a) $e_i=\{v_i v_{i+1}\}$  in the undirected case, and (b) $e_i=(v_i,v_{i+1})$  in the directed case, where in both cases the subscripts are taken modulo~$\ell$, i.e., $v_{\ell+1}=v_1$.

An undirected graph~$G=(V,E)$ is a \emph{tree} if it is \emph{acyclic}, i.e., it does not contain a cycle, and additionally, it is not possible to add an edge to~$G$ so that it remains acyclic. 
A \emph{clique} is a set of vertices in~$G$ whose vertices are pairwise adjacent, 
and an \emph{independent set} in~$G$ is a set of vertices that are pairwise non-adjacent.
The subgraph of~$G$ \emph{induced by} a set~$U$ of vertices has~$U$ as its vertex set and contains all edges of~$G$ with both endpoints in~$U$.
An \emph{arborescence} is a directed rooted tree with a unique path leading from the root to each non-root vertex.

\subsection{Hypergraph Classes}
\label{sec:app-hyp}
Here we provide Figure~\ref{fig:hypergraph_hierarchy} that describes the hierarchy of the various hypergraph classes studied in the paper.
We also prove 
a folklore fact about the relationship between laminar and subpath hypergraphs.

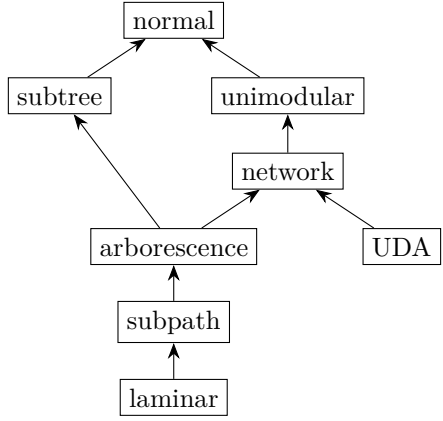
\begin{figure}[h]
    \centering
    \begin{tikzpicture}
\node[draw,rectangle] (N) at (0,0) {normal};
\node[draw,rectangle] (S) at (-1.5,-1) {subtree};
\node[draw,rectangle] (TU) at (1.5,-1) {unimodular};
\node[draw,rectangle] (NW) at (1.5,-2) {network};
\node[draw,rectangle] (A) at (0,-3) {arborescence};
\node[draw,rectangle] (RC) at (0,-4) {\rcp\ };
\node[draw,rectangle] (CC) at (3,-3) {UDA};
\node[draw,rectangle] (L) at (0,-5) {laminar};

\draw[-{Stealth[length=6pt]}] (S)--(N);
\draw[-{Stealth[length=6pt]}] (TU)--(N);
\draw[-{Stealth[length=6pt]}] (NW)--(TU);
\draw[-{Stealth[length=6pt]}] (A)--(S);
\draw[-{Stealth[length=6pt]}] (A)--(NW);
\draw[-{Stealth[length=6pt]}] (RC)--(A);
\draw[-{Stealth[length=6pt]}] (CC)--(NW);
\draw[-{Stealth[length=6pt]}] (L)--(RC);
    \end{tikzpicture}
    \caption{Connections between various classes of hypergraphs. An arrow from~$\mathcal{A}$ to~$\mathcal{B}$ means that hypergraphs with property~$\mathcal{A}$ are a subset of hypergraphs with property~$\mathcal{B}$. For the normality of subtree and unimodular hypergraphs, see~\cite[Section 4.4]{bretto-hypergraphs} and~\cite[Page~163]{berge-hypergraphs}, resp. In the figure, UDA stands for the class of hypergraphs that arise in our \textsc{University Dual Admission} problem, formally defined in Definition~\ref{def:uda-hypergraph}.    
}
    \label{fig:hypergraph_hierarchy}
\end{figure}

\begin{proposition}
\label{prop:laminar_RC1P}
    Laminar hypergraphs are \rcp\ hypergraphs.
\end{proposition}
\begin{proof}
    Consider a laminar hypergraph $\HH=(V,\EE)$. We say that a vertex~$v \in V$ \emph{belongs} to a given edge~$e \in \EE$ if $v \in e$ but there is no edge $f\subset e$, $f\in \EE$ that contains~$v$. 
    We use a recursive procedure called $\mathsf{List}(\HH)$ to list the vertices of~$\HH$ in an ordering that all hyperedges contain consecutive vertices. From that it is easy to see that the hyperedges indeed correspond to some subpaths of a directed path over $V$.
    \begin{itemize}[leftmargin=1.8cm]
        \item[Step 1.] Take a maximal edge~$e_{\max}$ in~$\HH$, and list all vertices belonging to~$e_{\max}$.
        \item[Step 2.] Let $\HH'$ be the hypergraph obtained restricting~$\HH$ to all vertices that are contained in~$e_{\max}$ but do not belong to~$e_{\max}$. \\
        Perform the recursive call $\mathsf{List}(\HH')$.
        \item[Step 3.] Let $\HH''$ be the subhypergraph of~$\HH$ induced by all vertices that are not in~$e_{\max}$. \\
        Perform the recursive call $\mathsf{List}(\HH'')$.
    \end{itemize}

    Since $\HH$ is laminar, we know that $\HH$ is the disjoint union of the hypergraph obtained by adding~$e_{\max}$ to~$\HH'$ and hypergraph~$\HH''$. Hence, the procedure lists each vertex of~$\HH$ exactly once.
    
    Consider any edge~$e \in \EE$, and observe that the first moment when a vertex of~$e$ is listed must be at the point where $e$ is a maximal edge in Step~1. Then the procedure lists all vertices belonging to~$e$, and 
    then  calls $\mathsf{List}(\HH'_e)$  where $\HH'_e$ contains exactly those edges in~$\EE \setminus \{e\}$ that are subsets of~$e$. 
    This means that during the call $\mathsf{List}(\HH'_e)$, including all of its recursive calls, exactly those vertices are listed that are contained in~$e$ but do not belong to~$e$. 

    Since no subset of~$e$ is taken as the maximal edge in Step~1 before $e$ is taken, we know that the vertices listed up to this moment are not in~$e$.
    This shows that the ordering of~$V$ created by $\mathsf{List}(\HH)$ yields an ordering, where~$e$ contains consecutive vertices. As this holds for each edge~$e \in \EE$, we obtain that $\HH$ is an  \rcp\ hypergraph.
\end{proof}

\section{Additional Material for Section~\ref{sec:uda}}
\label{sec:app-uda}

We start off by describing an instance, where a Gale-Shapley type proposal-rejection algorithm cycles.

\begin{Ex}
Consider an instance of \uda\ with students $s_1,s_2,s_3,s_4$, universities $u_1,u_2$ and programs $p_{11},p_{12},p_{13},p_2=p_{21}$. 
Let $u_1$ have capacity~2,~$u_2$ have capacity~1, and let  every program have quota~1. The preferences are as follows.
\begin{center}
\begin{tabular}{r@{\hspace{3pt}}lr@{\hspace{3pt}}lr@{\hspace{3pt}}l}
$s_1:$ & $p_{11}$ & $p_{11}: $ & $s_1\succ s_3  \quad$ & $u_1:$ & $s_4\succ s_3\succ s_2\succ s_1$\\
$s_2:$ & $p_{12}$ & $p_{12}:$ & $s_2$ & $u_2:$ & $s_3\succ s_4$\\
$s_3:$ & $p_{11}\succ p_2 \quad$ & $p_{13}:$ & $s_4 $ & &\\
$s_4:$ & $p_2\succ p_{13}$ & $p_2:$ & $s_3\succ s_4$\\

\end{tabular}
\end{center}
We consider a variant of the Gale-Shapley algorithm where students propose, and programs as well as universities may accept or reject these proposals.
Let the proposing order of the students be $s_3,s_4,s_1,s_2$. Then in the first run of the Gale-Shapley algorithm, the proposals are $s_3\to p_{11}$, 
$s_4\to p_2$, 
$s_1\to p_{11}$ (leading to the rejection of~$s_3$), 
$s_2\to p_{12}$, 
$s_3\to p_2$ (leading to the rejection of~$s_4$), 
$s_4\to p_{13}$ 
 (leading to the rejection of~$s_1$ by~$u_1$). After this first round of proposals, the current assignment is
\[\M=\{(s_2,u_1,p_{12}),(s_3,u_2,p_2),(s_4,u_1,p_{13})\}\] 
which still admits a blocking triple $(s_3,u_1,p_{11})$. 

Thus, let us  run the Gale-Shapley algorithm again, or more precisely, let us start another round of proposals where each student starts with~$\M$ as her current allocation, but the proposal pointers are reset. Then the proposals of the second round are: 
$s_3\to p_{11}$ (leading to the rejection of~$s_2$ by~$u_1$), $s_4\to p_2$, 
$s_1\to p_{11}$ (leading to the rejection of~$s_3$), 
$s_2\to p_{12}$,
$s_3\to p_2$ (leading to the rejection of~$s_4$), 
$s_4\to p_{13}$ (leading to the rejection of~$s_1$ by~$u_1$). 
Hence, after the second iteration we arrive at the same (unstable) assignment~$\M$. Hence, this variant of the Gale--Shapley algorithm never terminates on this instance. 
\end{Ex}

Intuitively, the problem that makes \uda\ difficult is the following. Assume for simplicity that programs have capacity one. While a university~$u_i$ is unsaturated, then in case there is a demand for the same program by at least two students, then since we do not know in advance whether the university will be filled at the end, we must let the program decide on whom to reject. However, this can lead to rejecting the worse student for~$u_i$. This can pose a problem, because if~$u_i$ gets saturated later and there is a demand for a free program at~$u_i$ by a student that $u_i$ prefers, then we have to reject a student from some program~$p_{ik}$. This can lead to some student $s_j$ creating a blocking triple $(s_j,u_i,p_{ik})$, because as we explained earlier, $s_j$ might have been necessarily rejected by the program. Finally, letting students propose again can easily lead to universities getting unsaturated again and thus to the algorithm cycling. 

\medskip
For the rest of the section, we first present all proofs that were omitted from Section~\ref{sec:uda} in Section~\ref{sec:app-uda-proofs}, and 
then proceed in Section~\ref{sec:app-uda-ILP} with an ILP formulation for the \uda\ problem. Then,  
we describe the notion of \emph{half-stability} and provide an algorithm to find a half-stable matching
in Section~\ref{sec:half-stab}.
Finally, we conclude in Section~\ref{sec:uda-allone} by showing that \uda\ is polynomial-time solvable if each university has capacity one.

\subsection{Proofs omitted from Section~\ref{sec:uda}}
\label{sec:app-uda-proofs}

\subsubsection*{Proof of Theorem~\ref{thm:uda-reduces-to-shm}}
\label{sec:app-proof-uda-reduces-to-shm}
\thmudareducestoshm*
\begin{proof}
Suppose that $\M$ is a stable assignment for~$I$. We define a  corresponding hypergraph $\capac^{\HH}$-matching as $M=\{\{s_j,u_i,p_{ik}\}:s_j \in S,p_{ik}=\M(s_j) \neq \nomatch\}$. To see that $M$ is indeed a stable $\capac^{\HH}$-matching for~$\HH_I$, assume for the sake of contradiction that a hyperedge $e=\{ s_j,u_i,p_{ik}\}$ blocks~$M$. This means that for each agent~$a$ incident to~$e$, either~$a$ is unsaturated, or prefers~$e$ to some hyperedge~$e' \in M$ incident to~$a$. By the definition of~$\succ^{\HH}$, this implies that 
(i) $s_j$ prefers $p_{ik}$ to $\M(s_j)$ (this allows $\M(s_j)=\nomatch$), 
(ii) $u_i$ is unsaturated in~$\M$, 
or there is a student $s_{j'}\in \M(u_i)$ such that $s_j\succ_{u_i}s_{j'}$, or $s_j\in \M(u_i)$ (recall that $u_i$ ranks hyperedges according to its preferences over the students, and ranks two hyperedges containing the same student according to the student's preferences), and 
(iii) $p_{ik}$ is either unsaturated or there is a student $s_{j''}\in \M(p_{ik})$ such that $s_j\succ_{p_{ik}} s_{j''}$. Hence, $e$ satisfies the conditions in Definition~\ref{def:uda-stability} for blocking~$\M$, contradiction. 

For the other direction, assume that $M$ is a stable hypergraph $\capac^{\HH}$-matching for~$\HH_I$. Define the corresponding assignment~$\M$ for~$I$ by setting $\M(s_j)=p_{ik}$ if $M$ contains a hyperedge $\{ s_j,u_i,p_{ik}\}$, and setting~$\M(s_j)=\nomatch$ if no hyperedge in~$M$ contains~$s_j$; recall that $\capac^{\HH}(s_j)=1$, so there can be at most one hyperedge in~$M$ incident to~$s_j$, so the assignment~$\M$ is well-defined. By the feasibility of~$M$, we know that $\M$ is an assignment that respects the capacities of the universities and the quotas of programs.
It remains to show that a blocking triple for the assignment~$\M$ would also be a blocking hyperedge for~$M$, by our definition of the extended preferences~$\succ^{\HH}$. So assume that the triple $ (s_j,u_i,p_{ik})$ blocks the assignment~$\M$.

Note first that since it blocks~$\M$, either $\M(s_j)=\nomatch$, or the preferences~$\succ^{\HH}_{s_j}$ are such that $s_j$ prefers~$e=\{ s_,,u_i,p_{ik}\}$ to the hyperedge $e'=\{ s_j,u_{i'},\M(s_j)\} \in M$ where $u_{i'}$ is the university offering the program~$\M(s_j)$.
Hence, the condition on~$s_j$ in 
Definition~\ref{def:shm-stability} necessary for~$e$ to be a blocking hyperedge for~$M$ is satisfied.
Second, if the university $u_i$ is unsaturated or $\M(u_i)$ contains a student~$s_{j'}$ for which $s_j \succ_{u_i} s_{j'}$, then the condition on~$u_i$ in 
Definition~\ref{def:shm-stability} necessary for~$e$ to block~$M$ is satisfied as well, because
$e \succ^{\HH}_{u_i} e''$ holds for the hyperedge $e''=\{ s_{j'},u_i,\M(s_{j'})\} \in M$.
For the case when $s_j \in \M(u_i)$, recall that we know $p_{ik} \succ_{s_j} \M(s_j)$, 
so by construction we get $e \succ^{\HH}_{u_i} \{s_j,u_i,\M(s_j)\} \in M$. Thus, $u_i$ fulfills the necessary condition of
Definition~\ref{def:shm-stability} for~$e$ to block~$M$ in all three cases.
Third, the program~$p_{ik}$ is either unsaturated, or $e \succ^{\HH}_{p_{ik}} e'''$ for some hyperedge $e'''=\{ s_{j''},u_i,p_{ik}\} \in M$ where $s_{j''}$ is the student whose existence is guaranteed by the third condition of Definition~\ref{def:uda-stability}, using that the triple~$(s_j,u_i,p_{ik})$ blocks~$\M$. 
Therefore, we can conclude that $e$ as a hyperedge blocks~$M$, a contradiction proving that~$\M$ must be stable. 
\end{proof}

\subsubsection*{Proof of Theorem~\ref{thm:max-uda}}
\label{sec:app-proof-max-uda}

\thmmaxuda*
\begin{proof}
We reduce from the  \textsc{com-smti} problem, the problem of deciding whether an instance of \textsc{stable matching with ties and incomplete preferences} admits a complete stable matching. The input instance~$I$ of \textsc{com-smti} contains a set $V_w=\{w_1,\dots,w_n\}$ of women, a set~$V_m=\{m_1,\dots, m_n\}$ of men, and a preference list~$\succ_a$ for each person $a \in V_w \cup V_m$, which is a weak ordering over a subset of persons from the opposite sex. We will assume that each man~$m_i$ has a strict preference list, and each woman~$w_j$ has either a strict preference list, or her preference list is a tie of length two; it is known that \textsc{com-smti} remains $\NP$-hard even in this special case~\cite{MIIMM02}. We may assume w.l.o.g.\ that the set of women whose preference list is a tie is $\{ w_1,\dots,w_{\ell} \}$.

Create an instance $I'$ of \uda\ as follows.
For each woman $w_i \in V_w$ we create a university $u_i$ along with a program $p_i=p_{i1}$, where both the capacity of~$u_i$ and the quota of~$p_i$ are~$1$. 
For each man $m_j \in V_m$, we create a student~$s_j$. 
The preference list of each student $s_j$ is inherited from $m_j$, by replacing each woman~$w_i$ in the preference list of~$m_j$ with the program $p_j$. 
For some $i \in [\ell]$, there are two acceptable triples containing university~$u_i$ and program~$p_i$, namely $(s_k,u_i,p_i)$ and~$(s_l,u_i,p_i)$ where $m_k$ and~$m_l$ ($k<l$) are the two acceptable partners of~$w_i$. The preferences of~$u_i$ and~$p_i$ ($i \in [\ell]$) are such that $s_k \succ_{u_i} s_l$ and $s_l\succ_{p_i} s_k$. For the remaining universities and programs, i.e., for $i \in [n ] \setminus [\ell]$, the preference lists of~$u_i$ and~$p_i$ are the same and are inherited from $w_i$, by replacing each man~$m_j$ in the preference list of~$w_i$ with the student~ $s_j$. We set each capacity to be~$1$ and the desired size of a stable assignment for~$I'$ as $t=n$.

We claim that there is a complete stable matching in~$I$ if and only if there is a stable assignment matching of size $t=n$ in~$I'$. 

Let $M$ be a complete stable matching in~$I$. Let $\M$ be the assignment obtained by setting $\M (s_j)=p_i$ for each $\{w_i,m_j\} \in M$. As $M$ is a matching covering each person exactly once, $\M$ is a feasible assignment of size~$n$ (in which no agent is unsaturated). Suppose that there is a blocking triple $(s_j,u_i,p_i)$ for~$\M$. Then $s_j$ prefers $p_i$ to $\M (s_j)=p_{i'}$, which means that $m_j$ prefers~$w_{i}$ to $w_{i'}=M(s_j)$. 
Similarly, $u_i$ and $p_i$ both prefer $s_j$ to the unique student assigned to $p_i$ in~$\M$. This implies that 
$u_i$ and~$p_i$ must correspond to a woman~$w_i$ whose preference list is strict, as otherwise $u_i$ and $p_i$ would have opposite rankings over the students who apply there and one of them would prefer $\M (p_i)$ to $s_j$. Thus, we obtain that $w_i$ prefers~$m_j$ to~$m_{j'}=M(w_i)$. This proves that 
$\{ w_i,m_j\}$ blocks $M$, a contradiction to the stability of~$M$. Thus, $\M$ is stable as required.

Suppose now that we have a stable assignment~$\M$ of size~$n$ in~$I'$. Observe that for each $i \in [n]$. Since all capacities in~$I'$ are~$1$, we can create a matching $M$ in~$I$ by adding $\{w_i,m_j\}$ whenever $\M (s_j)=p_i$. Clearly, $M$ has size~$n$, hence it is complete. Suppose that $\{ m_j,w_i\}$ blocks $M$. Then $w_i$ must have a strict preference list, and thus it is straightforward to check that the triple $(s_j,u_i,p_i)$ blocks~$\M$, a contradiction. Hence, $M$ is stable, as required.
\end{proof}

\subsection{An Integer Programming Formulation for \uda}
\label{sec:app-uda-ILP}

Let $I=(S,U,P,\capacuda,\capacquota,(\succ_a)_{a \in S \cup U \cup P})$ be an instance of \uda.
Without loss of generality, we may suppose that each program has quota one, as we can replace each program~$p$ with~$q(p)$ clones $p^1,\dots,p^{q(p)}$, each having capacity~$1$ and retaining the preferences of~$p$; we then replace $p$ with the list~$p^1,\dots,p^{q(p)}$ in the preference lists of students. 

 For convenience, we rely on the corresponding \tushbm\ instance $I'=(\HH_I,b^{\HH},(\succ_a^{\HH})_{a\in S\cup U\cup P})$. Let $\HH_I=(S \cup U \cup P,\EE_I)$ be the hypergraph associated with~$I$, and let $|\EE_I|=m$.

\leqnomode
\begin{align}
\tag{IP$_{\textup{UDA}}$} 
\label{IP:UDA}
\phantom{a}\\[-20pt] \notag
\sum_{e\ni s}x_{e} & \le 1 \quad & \forall s\in S \\ \notag
   \sum_{e\ni u}x_{e} & \le \capacuda(u)  \quad 
   & \forall u\in U \\ \notag
   \sum_{e\ni p}x_{e} & \le 1 \quad & \forall p\in P \\ \notag
    \capacuda(u)\bigg(\sum_{e'\succ^{\HH}_se}x_{e'} +\sum_{e'\succ^{\HH}_pe}x_{e'}+x_e\bigg)+\sum_{e'\succ^{\HH}_ue}x_{e'} & \ge \capacuda(u) \quad 
    & \forall e=(s,u,p)\in \EE_I \\ \notag
    x\in \{ 0,1\}^m  & &
\end{align}

Consider an arbitrary solution~$x$ to~\ref{IP:UDA}.
Clearly, the first three inequalities in~\ref{IP:UDA} ensure that $x$ corresponds to a $\capac^{\HH}$-matching~$M_x$ in~$I'$ whose characteristic vector is~$x$.
By the last inequality, for each hyperedge $e=\{ s,u,p\}$ it holds that either $e \in M_x$, or there is a hyperedge~$e' \in M$ that preferred to~$e$ by either~$s$ or~$p$, or $\sum_{e'\succ_ue}x_{e'}\ge \capacuda (u)$, so $u$ is saturated in~$M_x$ with hyperedges that it prefers to~$e$. Hence, $M_x$ is stable. 

Using the above observations, it is straightforward to verify that the stable matchings for~$\HH_I$ 
are exactly the solutions of~\ref{IP:UDA}. Therefore, as we have seen, this IP is guaranteed to have a solution.

By Theorem~\ref{thm:uda-reduces-to-shm}, $M_x$ also gives a stable assignment $\M_x$ to the \uda\ instance $I$.

\subsection{A weaker version of stability for \uda}\label{sec:half-stab}

We can define stability in \uda\ in a slightly weaker, but still natural way. To avoid confusion with previous notions of weak and strong stability in the literature, we will call this relaxed notion \emph{half-stability}.

\begin{Def}
A feasible assignment $\M$ is \textit{half-stable} if there is no $(s_j,u_i,p_{ik})$ triple, such that 
\begin{itemize}
    \item [--]$(s_j,u_i,p_{ik})$ blocks $\M$ and
    \item[--] if $p_{ik}$ is saturated, then there is a student $s_{j'} \in \M(p_{ik})$, who is hence also in $\M(u_i)$, such that $s_j\succ_{u_i} s_{j'}$ and $s_j\succ_{p_{ik}} s_{j'}$.
 \end{itemize}
 If there is such a triple, then we call it \textit{doubly blocking}. 
\end{Def}
The motivation behind this definition is that if there is a blocking triple $(s_j,u_i,p_{ik})$, and $p_{ik}$ is saturated, then in order to accept $s_j$ to $p_{ik}$, we have to drop a student from~$p_{ik}$, hence also from $u_i$. However, if there is no student that is worse for both $u_i$ and $p_{ik}$, either $u_i$ or $p_{ik}$ would have to loose a better student in exchange for $s_j$, so they may prefer not to participate in the blocking after all. 

Of course, a stable assignment $\M$ is also \textit{half-stable}, hence a half-stable assignment always exists. We show that one can also be found in polynomial time.

\begin{theorem}
\label{thm:uda-halfstability}
    For an instance of \uda\, a half-stable assignment can be found in $\mathcal{O}(|\EE|)$ time where $\EE$ is the set of acceptable triples.
\end{theorem}
\begin{proof}
We reduce this problem to the  \textsc{Student--Project Allocation} (or \textsc{spa}) problem which can be solved in linear time due to the results of Abraham, Irving and Manlove~\cite{AIM07}. This problem can be thought of as the special case of \uda\ where the preference ordering of each program~$p$ coincides with the preference ordering of the university offering~$p$. 

Given an instance $I$
of \uda, 
we construct an instance~$I'$ of \textsc{spa} by simply setting the preferences of each program~$p_{ik}$ to be identical to~$\succ_{u_i}$. A stable matching~$M$ in~$I'$, which can be found in linear time~\cite{AIM07}, 
is automatically a half-stable assignment $\M$ in~$I$ as well.

To see this, assume for the sake of a contradiction that
$(s_j,u_i,p_{ik})$ doubly blocks~$\M$. 
Then $s$ prefers $p_{ik}$ to $\M(s_j)$ (allowing the possibility that $\M(s_j)=\nomatch)$. 
By the stability of~$M$ in~$I'$, this means that either the program~$p_{ik}$ or the university~$u_i$ was saturated by students that were better than~$s_j$ according to~$\succ_{u_i}$. If~$u_i$ is saturated by students better than~$s_j$, then this contradicts our assumption that $(s_j,u_i,p_{ik})$ (doubly) blocks~$\M$ in~$I$. 
Thus, it must be the case that $p_{ik}$ is saturated with students that $u_i$ prefers to~$s_j$; however, then there is no student in $\M(p_{ik})$ that is worse than~$s_j$ for both~$p_{ik}$ and~$u_i$, contradicting our assumption that $(s_j,u_i,p_{ik})$ doubly blocks~$\M$. 
\end{proof}

\subsection{Solving \uda\ with unit capacities}\label{sec:uda-allone}

Consider an instance $I=(S,U,P,\capacuda, \capacquota, (\succ_v)_{v\in S\cup U \cup P})$, where $\capacuda (u)=1$ for all universities $u\in U$. 
We show that the following simple version of the deferred acceptance algorithm by Gale and Shapley finds a stable matching for~$I$. 

Initially, every student is unassigned. While there exists an unassigned student, we do the following.
Let $s_j$ be an unassigned student, and let $ p_{ik}$ be the best program for $s_j$ that has not rejected $s_j$ yet. We let $s_j$ propose to $p_{ik}\in P_i$. However, the university will be the one who decides whether $s_j$ is rejected or not. That is, if there are no students at $u_i$, then $s_j$ gets assigned to $p_{ik}$. If there is a student $s_{j'}$ such that $s_{j'}\succ_{u_i}s_j$, then $s_j$ is rejected. Finally, if $u_i$ has a student $s_{j'}$, but he is worse than $s_j$, then $s_j$ gets assigned to $p_{ik}$ and $s_{j'}$ gets rejected.

Let the assignment thus obtained be $\M$; we claim that $\mu$ is stable. Suppose for the contrary that a triple $(s,u,p)$ blocks $\M$. Then, $p\succ_s \M (p)$, so $s$ has proposed to $p$, but got rejected. Hence, at that point $u$ had a better student~$s'$ assigned. Since no student is rejected for another student that is worse for the university, it follows that $u$ has a better student assigned in $\M$ than $s$. Finally, as $\capacuda (u)=1$, this implies that $u$ is saturated with better students than $s$, a contradiction to $(s,u,p)$ blocking~$\mu$.

\medskip
The intuitive reason why the above algorithm works is that for unit capacities stability coincides with half-stability. 
Indeed, if $(s_j,u_i,p_{ik})$ is a triple that blocks a matching~$\mu$ for~$I$ and $p_{ik}$ is saturated by~$\mu$, then by $\capacuda(u_i)=1$ there must be a unique student~$s_{j'}$ assigned to~$p_{ik}$ and to~$u_i$, and since $(s_j,u_i,p_{ik})$ blocks~$\mu$, it immediately follows that $s_j \succ_{u_i} s_{j'}$ and $s_j \succ_{p_{ik}} s_{j'}$, i.e.,  the triple $(s_j,u_i,p_{ik})$ doubly blocks~$\mu$. This means that half-stability implies stability and, hence, Theorem~\ref{thm:uda-halfstability} has the following consequence.

\begin{Cor}
    For an instance of \uda\ with unit capacities, a stable assignment can be found in $O(|\EE|)$ time where $\EE$ is the set of acceptable triples.
\end{Cor}

\section{Additional Material for  Section~\ref{sec:algorithms}}

We first present all omitted proofs from Sections~\ref{sec:laminar} and~\ref{sec:rc1p} in Sections~\ref{sec:app-laminar-proofs} and~\ref{sec:app-rc1p-proofs}, respectively.

\subsection{Omitted proofs from Section~\ref{sec:laminar}}
\label{sec:app-laminar-proofs}

\subsubsection*{Proof of Theorem~\ref{thm:laminar-weighted-nphard}}
\label{sec:app-proof-laminar-weighted-nphard}

\thmrcpconstrainednphard*
\begin{proof}
It is clear that the problem is in $\NP$. To prove its $\NP$-hardness, 
we present a reductions from the {\sc cnf-sat} problem.
The input of {\sc cnf-sat} is a CNF formula~$\varphi$ defined over variables~$x_1, \dots, x_n$ 
and containing clauses $c_1,\dots,c_m$; the task is to decide whether $\varphi$ is satisfiable.

\medskip
\noindent
{\bf Construction.}
We are going to define a hypergraph~$\HH$ with vertex set $V=\{z\} \cup \{v_i,\bar{v}_i : i \in [n]\}$
and edge set $\EE=\{e_z\} \cup E_V \cup F$ where $E_V=\{e_i,\bar{e}_i : i \in [n]\}$ 
and $F=\{f_j: j \in [m]\}$. 
For convenience, we allow our hypergraph~$\HH$ to contain parallel edges, 
so we will use an incidence function $\psi: \EE \rightarrow 2^V$ to define~$\HH$ as follows:

\begin{tabular}{ll}
$\psi(e_z)=\{z\}$;  \\
$\psi(e_i)=\psi(\bar{e}_i)=\{v_i,\bar{v}_i\}$ & for each $i \in [n]$; \\
$\psi(f_j)=V $ & for each $j \in [m]$. \\
\end{tabular}

We define the edge sets $F_i^+=\{f_j:c_j \textrm{ contains the literal $x_i$}\}$ and 
$F_i^-=\{f_j:c_j \textrm{ contains the literal $\bar{x}_i$}\}$ for each variable~$x_i$.
The preferences associated with the vertices of~$\HH$ are defined below; a set in the preferences denotes any arbitrarily fixed ordering of the set.

\begin{tabular}{l}
$z:F \succ e_z$;  \\
$v_i:F \setminus F^+_i \succ e_i \succ F^+_i \succ \bar{e}_i$;  \\
$\bar{v}_i:F \setminus F^-_i \succ \bar{e}_i \succ F^-_i \succ e_i$.
\end{tabular}

We set the weight function~$\weight$ to be~$0$ on every hyperedge except for~$e_z$, and we set $\weight(e_z)=1$. Moreover, we set our threshold (the weight of the desired stable matching) as~$t=1$.
This finishes our instance~$I_{\HH}$ of \maxwlamshm. 
It is clear that the hypergraph~$\HH$ %
is laminar, as promised.

\medskip
\noindent
{\bf Correctness.}
Observe that weight of a stable matching reaches the threshold~$t$ in $I_{\HH}$ if and only if it contains the hyperedge~$e_z$.
We claim that $I_{\HH}$  admits a stable hypergraph containing $e_z$ if and only if 
$\varphi$ can be satisfied.

First assume that a stable matching~$M$ in~$I_{\HH}$ contains~$e_z$. 
Clearly, $M \cap F=\emptyset$, since every edge in~$F$ is incident to~$z \in \psi(e_z)$. 
Moreover, $|M \cap \{e_i,\bar{e}_i\}|=1$ for each $i \in [n]$, as otherwise both $v_i$ and $\bar{v}_i$ would be unmatched (by $M \cap F=\emptyset$), 
and thus $e_i$ (and also $\bar{e}_i$) would block~$M$. 
We construct a truth assignment~$\alpha_M$ for~$\varphi$ by setting $\alpha_M(x_i)=\texttt{true}$ if $e_i \in M$, 
and $\alpha_M(x_i)=\texttt{false}$ if~$\bar{e}_i \in M$. 

To see that $\alpha_M$ satisfies~$\varphi$, let us consider any clause $c_j$. 
Since $z$ prefers $f_j$ to~$e_w$, but $f_j$ cannot block~$M$, 
there must be some vertex~$v \in \psi(f_j) \setminus \{z\}$ 
and an edge~$e \in M$ incident to~$v$ that is preferred by $v$ to~$f_j$.
First, if $v=v_i$ for some~$i$, then the edge of~$M$ preferred by~$v_i$ to~$f_j$ must be $e_i$ (since $\bar{e}_i$ is the edge least preferred by~$v_i$).
Hence, $\alpha(x_i)=\texttt{true}$. Moreover, since $v_i$ prefers $e_i$ to $f_j$, it must be that $f_j \in F_i^+$, that is, 
the clause $c_j$ contains $x_i$ as a positive literal. Thus, $\alpha_M$ makes $c_j$ true.
Similarly, if~$v=\bar{v}_i$ for some~$i$, then the edge of~$M$ preferred by~$\bar{v}_i$ to~$f_j$ must be $\bar{e}_i$ 
(since $e_i$ is the edge least preferred by~$\bar{v}_i$).
Hence, $\alpha(x_i)=\texttt{false}$. Moreover, since $\bar{v}_i$ prefers $\bar{e}_i$ to $f_j$, it must be that $f_j \in F_i^-$, that is, 
the clause $c_j$ contains the literal $\bar{x}_i$. Thus, $\alpha_M$ again makes $c_j$ true.
Since this holds for each clause~$c_j$ of~$\varphi$, we get that $\alpha_M$ satisfies~$\varphi$.

For the other direction, assume now that $\alpha$ is a truth assignment that satisfies~$\varphi$. 
Let \[M_\alpha=\{e_z\} \cup \{e_i: \alpha(x_i)=\texttt{true}\} \cup \{\bar{e}_i:\alpha(x_i)=\texttt{false}\}.\] 
Clearly, the edges in~$M_\alpha$ cover each vertex of~$V$ exactly once, so $M_\alpha$ is indeed a matching. 
To see that $M_\alpha$ is stable as well, first observe that no edge of the form~$e_i$ or~$\bar{e}_i$ can block~$M_\alpha$, 
as any such edge is the least-preferred edge by some vertex in~$V$. Therefore we only have to prove that no edge in~$F$ blocks~$M_\alpha$. 
Consider some $f_j \in F$. Since $c_j$ contains a literal that is set to $\texttt{true}$ by~$\alpha$, 
there must exist either some $e_i \in M_\alpha$ for which $c_j \in F_i^+$, 
or some $\bar{e}_i \in M_\alpha$ for which $c_j \in F_i^-$. 
In the former case, $f_j$ does not block~$M_\alpha$ because $v_i$ prefers $e_i \in M_\alpha$ to $f_j \in F^+_i$, 
and in the latter case, $f_j$ does not block~$M_\alpha$ because $\bar{v}_i$ prefers $\bar{e}_i \in M_\alpha$ to $f_j \in F^-_i$. 
This proves that $M_\alpha$ is stable in~$I_\HH$ and our reduction is correct. 
\end{proof}

\subsection{Omitted proofs from Section~\ref{sec:rc1p}}
\label{sec:app-rc1p-proofs}

\subsubsection*{Proof of Lemma~\ref{lem:subpath_real}} 
\label{sec:app-proof-subpath_real}
\subpathlemmareal*
\begin{proof}
We use induction on~$i$. The case $i=0$ is trivial, so assume that the lemma holds for $i-1$.

First, suppose that $(\sigma,M)$ was added to~$S_i$ on line~\ref{alg2:lineX}, so $M \subseteq \HE_{i-1}$.
Let $(\sigma^-,M) \in S_{i-1}$ be the pair selected on line~\ref{alg2:lineA}
in the cycle during which $(\sigma,M)$ was put into~$T_i$ on line~\ref{alg2:lineB}. 
By induction, %
\begin{itemize}
    \item[(a)] $M$ is compatible with~$\sigma^-$,
    \item[(b)] $\sigma^-(v) \in M(v)$ for every $v \in L_{i-1}$ where $\sigma^-(v) \in \HE_{i-1}$, and 
    \item[(c)] every edge in $\EE_{i-1} \setminus M$ that blocks~$M$ is planned to be dominated by~$\sigma^-$.
\end{itemize}

First, observe that the $\capac$-matching~$M$ is compatible with~$\sigma$: the compatibility conditions for each $v \in L_{i-1} \cap L_i$ follow immediately from $\sigma^-(v)=\sigma(v)$ due to~(a), while the compatibility condition for each $v \in L_i \setminus L_{i-1}$ 
holds irrespective of the value of~$\sigma(v)$, because edges of~$M \subseteq \HE_{i-1}$ have no vertices in~$L_i \setminus L_{i-1}$. 

Second, by our assumption that 
Algorithm~\ref{alg:RC1P} added $(\sigma,M)$ to~$S_i$ on line~\ref{alg2:lineX}, we know that $\sigma(v) \neq e_i$ for each $v$ incident to~$e_i$. 
Hence, if $\sigma(v) \in \EE_i$ for some $v \in L_i$, then $\sigma(v) \in \EE_{i-1}$, and consequently, $v \in L_i \cap L_{i-1}$ (because vertices of $L_i \setminus L_{i-1}$ are not contained in any edge of~$\EE_{i-1}$). 
Therefore, (b) implies
$\sigma(v)=\sigma^-(v) \in M(v)$. 
Third, it remains to show that every edge $e \in \HE_i \setminus M$ that blocks $M$ is planned to be dominated by~$\sigma$.
First note that Algorithm~\ref{alg:RC1P} checks that this holds for~$e_i$ explicitly on line~\ref{alg2:lineX-1}, so we may assume $e \neq e_i$.
By (c), we know that every edge in~$\HE_{i-1} \setminus M$ that blocks~$M$ is planned to be dominated by~$\sigma^-$, 
so there is some $v \in L_{i-1}$ such that $\sigma^-(v) \succ_v e$. On the one hand, if $v \in L_i \cap L_{i-1}$, then $e$ is planned to be dominated by $\sigma$ too, because $\sigma(v)=\sigma^-(v)$.
On the other hand, if $v \in L_{i-1} \setminus L_{i}$, then Algorithm~\ref{alg:RC1P} must have confirmed on line~\ref{alg2:lineB}
that $M$ is complete on~$v$ (otherwise it would not have put $(\sigma^-,M)$ into~$T_i$),  which implies that $M$ dominates~$e$ at~$v$. Thus, $e$ does not block~$M$, proving 
that every edge in $\EE_i \setminus M$ that blocks~$M$ is planned to be dominated by~$\sigma$.
This finishes the proof that $M$ realizes~$\sigma$ with respect to~$\HE_i$.

Suppose now that $(\sigma,M)$ was put into~$S_i$ on line~\ref{alg2:lineY}.
Let $M^-=M \setminus \{e_i\}$, and let $(\sigma^-,M^-) \in S_{i-1}$ be the pair selected on line~\ref{alg2:lineA}
in the cycle during which $(\sigma,M^-)$ was put into~$T_i$ on line~\ref{alg2:lineB}.  
By induction, conditions~(a)--(c) hold with $M^-$ replacing~$M$; let these modified conditions be referred to as (a')--(c').
First, by (a') and the conditions checked on line~\ref{alg2:lineY-2} we get that $M=M^- \cup \{e_i\}$ is a $\capac$-matching that is compatible with~$\sigma$ on~$L_i$.
Second, if $\sigma(v) \in \HE_i$ for some $v \in L_i$, then either $\sigma(v) \in \HE_{i-1}$, in which case $\sigma(v)=\sigma^-(v) \in M^-(v) \subseteq M(v)$ by (b'),
or $\sigma(v) =e_i \in M(v)$. 
Hence, it remains to see that every edge $e \in \HE_i \setminus M = \HE_{i-1} \setminus M^-$ that blocks $M$ is planned to be dominated by~$\sigma$.
First note that $e$ must block also $M^-$, and thus by (c') must be planned to be dominated by~$\sigma^-$ at some~$v \in L_{i-1}$. 
We can argue exactly as for previous case to prove that either 
$e$ is planned to be dominated by~$\sigma$ at $v$ (in case $v \in L_i$)
or $e$ is dominated by~$M$ at~$v$ (in case $v \notin L_i$). Thus, again it holds that $M$ realizes $\sigma$ with respect to~$\HE_i$. 
\end{proof}

\subsubsection*{Proof of Theorem~\ref{thm:RC1P-Whardness}}
\label{sec:app-proof-RC1P-Whardness}

\thmrcpwhardness*
\begin{proof}
We are going to present a reduction from {\sc Multicolored Clique} which is known to be $\mathsf{W}[1]$-hard~\cite{FHRV2009-multicolored-hardness}.
The input of {\sc Multicolored Clique} is a graph~$G=(V,E)$ and integer~$k$ with the vertex set of~$G$ is partitioned into $k$ independent sets~$V_1,\dots,V_k$;
the task is to decide whether $G$ admits a clique of size~$k$. 

\smallskip
\noindent
{\bf Construction.}
We first construct a \emph{selection gadget} which involves vertices $w,z,s_1,\dots,s_k$ and edge set $\{e_w, d_1, \dots, d_k\} \cup \{e^0(v):v \in V\}$. 
We will also need \emph{incidence gadgets} which will involve a vertex set $\{a_e,a'_e:e \in \bar{E}\}$ and edge set $\{f_e:e \in \bar{E}\}$ 
where $\bar{E}$ contains all edges in the complement of~$G$ running between different partitions. 
We fix an ordering over~$\bar{E}$, and denote its elements by $e_1,\dots, e_{\bar{m}}$. 
We will also need \emph{repeater gadgets} which will contain vertices $\{r_{(i,e)}^h:h \in [4], i \in [k], e \in \bar{E}\}$
and edge set $\{e^j(v),\hat{e}^j(v),\hh{e}^j(v):e_j \in \bar{E},i \in [k],v \in V_i\}$.

The incidence function~$\psi$ of the edges is as follows: 
\[\begin{array}{l@{\hspace{1pt}}l@{\hspace{1pt}}ll}
\psi(e_w)&=&\{w,z\}; \\
\psi(d_i)&=&\{z,s_1,\dots,s_i\} & \textrm{for $i \in [k]$}; \\
\psi(e^0(v))&=&\{s_i,\dots,s_k\} \cup \{r_{(i',e_1)}^h:h \in [4],i' \in [i-1]\} \\
&& \quad \cup \{r_{(i,e_1)}^1,r_{(i,e_1)}^2,r_{(i,e_1)}^3\}
& \textrm{for $v \in V_i$, $i \in [k]$;} \\
\psi(e^j(v))&=&\{r_{(i,e_j)}^2,r_{(i,e_j)}^3,r_{(i,e_j)}^4\} \cup \{r_{(i',e_j)}^h:h \in [4],i' \in [k] \setminus [i]\} \\
&& \quad \cup \{a_{e_j},a'_{e_j}\} 
\cup \{r_{(i',e_{j+1})}^h:h \in [4],i' \in [i-1]\} \\
&& \quad \cup \{r_{(i,e_{j+1})}^1,r_{(i,e_{j+1})}^2,r_{(i,e_{j+1})}^3\}
& \textrm{for  $v \in V_i$, $i \in [k]$, $j \in [\bar{m}-1]$;} \\
\psi(e^j(v))&=& \{r_{(i,e_j)}^2,r_{(i,e_j)}^3,r_{(i,e_j)}^4\} \cup \{r_{(i',e_j)}^h:h \in [4],i' \in [k] \setminus [i]\} \\
&& \quad \cup \{a_{e_j},a'_{e_j}\}
& \textrm{for $v \in V_i$, $i \in [k]$, $j=\bar{m}$}; \\
\psi(\hat{e}^j(v))&=&\{r_{(i,e_j)}^1,r_{(i,e_j)}^2\} & \textrm{for $v \in V$, $i \in [k]$, $j \in \bar{m}$;} \\
\psi(\hh{e}^j(v))&=&\{r_{(i,e_j)}^3,r_{(i,e_j)}^4\} & \textrm{for $v \in V$, $i \in [k]$, $j \in \bar{m}$;} \\
\psi(f_e) &=& \{a_e,a'_e \} & \textrm{for $e \in \bar{E}$}. 
\end{array}
\]
It is straightforward to check that the following ordering of the vertices in the resulting hypergraph~$\HH$ witnesses its \rcp\ property.
We start with the vertices of the selection gadget as $w,z,s_1,\dots,s_k$, and then 
proceed with a \emph{segment} for each $e_j$, $j=1,\dots,\bar{m}$, 
where the segment for~$e$ is defined as the list 
\[r_{(1,e)}^1,r_{(1,e)}^2,r_{(1,e)}^3,r_{(1,e)}^4,\dots, r_{(k,e)}^1,r_{(k,e)}^2,r_{(k,e)}^3,r_{(k,e)}^4,a_e,a'_e.\]

We set the capacities of the vertices as follows:
\[\begin{array}{ll}
\capac(w) = \capac(z) = 1; \\
\capac(s_i) = i & \textrm{for each $i \in [k]$;} \\
\capac(r_{(i,e)}^1)=\capac(r_{(i,e)}^4)=c(a_e)=\capac(a'_e)=k & \textrm{for each $i \in [k]$, $e \in \bar{E}$;} \\
\capac(r_{(i,e)}^2)=\capac(r_{(i,e)}^3)=k+1 & \textrm{for each $i \in [k]$, $e \in \bar{E}$.} 
\end{array}\]

\def\hype{\mathcal{E}}
Next, before defining the preferences of the vertices in~$\HH$, we need some additional notation.
For convenience, for each $j=0,1,\dots,\bar{m}$ and $i  \in [k]$ we define hyperedge sets $\hype_i^j=\{e^j(v):v \in V_i\}$,
$\hype_i^j=\{\hat{e}^j(v):v \in V_i\}$, and $\hh{\hype}_i^j=\{\hh{e}^j(v):v \in V_i\}$.
Let us also fix an arbitrary ordering over each $V_i$, and denote the resulting ordering of $V_i$ and the corresponding ordering of the three sets just defined 
by~$\underrightarrow{V}_i$, $\underrightarrow{\hype}_i^j$, $\underrightarrow{\hat{\hype}}_i^j$, and $\underrightarrow{\hh{\hype}}_i^j$, respectively.
The reverse ordering of these sets will be denoted 
by~$\underleftarrow{V}_i$, $\underleftarrow{\hype}_i^j$, $\underleftarrow{\hat{\hype}}_i^j$, and $\underleftarrow{\hh{\hype}}_i^j$, respectively.
Furthermore, we write $\hype_{<i}^j=\bigcup_{i' \in [i-1]} \hype_{i'}^j$ and $\hype_{>i}^j=\bigcup_{i' \in [k] \setminus [i]} \hype_{i'}^j$.
We let $\hype^j=\bigcup_{i \in [k]} \hype_i^j$, and we define $\hat{\hype}^j$ and $\hh{\hype}^j$ analogously.

We also define $\bar{E}_{i,i'}$ to contain those edges $xy$ of $\bar{E}$ for which $x \in V_i$ and $y \in V_{i'}$. 
For two series $A=(a_1,\dots, a_n)$ and $B=(b_1,\dots, b_n)$ of the same length
let us denote the merging of $A$ and $B$ as $A \merge B=(a_1,b_1,\dots,a_n,b_n)$.
Now we are ready to define the preferences, yielding an instance $I_\HH$ of \shbm;
again, sets in the preferences are ordered arbitrarily.
 
\[\begin{array}{r@{\hspace{3pt}}r@{\hspace{3pt}}ll}
w &:& e_w; \\
z &:& \{d_1,\dots,d_k\} \succ e_w; \\
s_i &:& \{d_{i+1},\dots,d_k\} \succ \hype_{<i}^0 \succ \hype_i^0 \succ d_i  & \textrm{for $i \in [k]$}; \\
r_{(i,e_j)}^1 &:& \hype_{<i}^{j-1} \cup \hype_{>i}^{j-1} \succ \underrightarrow{\hype}_i^{j-1} \merge \underrightarrow{\hat{\hype}}_i^j 
 & \textrm{for $i \in [k]$, $j \in [\bar{m}]$;}  \\
r_{(i,e_j)}^2 &:& \hype^{j-1} \succ \underleftarrow{\hat{\hype}}_i^j \merge \underleftarrow{\hype}_i^j
 & \textrm{for $i \in [k]$, $j \in [\bar{m}]$;}  \\
r_{(i,e_j)}^3 &:& \hype_{<i}^{j-1} \cup \hype_{>i}^{j-1} \cup  \hype_i^j \succ \underleftarrow{\hh{\hype}}_i^j \merge \underleftarrow{\hype}_i^{j-1} 
 & \textrm{for $i \in [k]$, $j \in [\bar{m}]$;}  \\
r_{(i,e_j)}^4 &:& \hype_{<i}^{j-1} \cup \hype_{>i}^{j-1} \succ \underrightarrow{\hype}_i^j \merge \underrightarrow{\hh{\hype}}_i^j 
 & \textrm{for $i \in [k]$, $j \in [\bar{m}]$;}  \\
a_{e_j} &:& \hype^j \setminus \{e^j(x)\} \succ f_{e_j} \succ e^j(x) 
 & \textrm{for $e_j=\{x,y\} \in \bar{E}_{i,i'}$, $i<i'$;}  \\
a'_{e_j} &:& \hype^j \setminus \{e^j(y)\} \succ f_{e_j} \succ e^j(y) 
 & \textrm{for $e_j=\{x,y\} \in \bar{E}_{i,i'}$, $i<i'$}.  \\
\end{array}\]

To finish our instance of {\sc \maxwshbm}, we set the weight of~$e_w$ as~1, and we set the weight of every other edge as~$0$; moreover, we set $t=1$ as the weight of the desired stable $\capac$-matching.  
Clearly, a $\capac$-matching has weight at least~$t=1$ if and only if contains the edge~$e_w$.

Notice that the maximum capacity is~$\capac_{\max}=k+1$, and the maximum size of any edge is $\ell_{\max}=4k+4$, 
the size of any edge $e^j(v)$ with $1 \leq j<\bar{m}$.
It remains to show that our reduction is correct, i.e., $e_w$ is contained in some stable $\capac$-matching of $I_\HH$ if and only if $G$ admits a clique of size~$k$.

\smallskip
\noindent
{\bf Direction ``$\Rightarrow$'':}
Suppose first that $M$ is a stable $\capac$-matching for~$I_\HH$ with $e_w \in M$. 
Clearly, any edge not in~$M$ must be dominated at some vertex.

Note that since $z$ prefers each of the edges $d_1, \dots, d_k$ to~$e_w$ and $\capac (z)=1$,
it must be the case for each $i \in [k]$ such that $d_i \notin M$ and $d_i$ is dominated at some vertex other than~$z$. 
By $\psi(d_1)=\{z,s_1\}$ we know that $d_1$ must be dominated at $s_1$; by $\capac(s_1)=1$ this implies $|\hype^0_1 \cap M|=1$. 
More generally, observe that $d_i$ for some $i \in [k]$ can only be dominated at $s_i$ 
(as the vertices $s_1,\dots,s_{i-1}$ do not prefer any edge that might be in~$M$ to~$d_i$), 
implying that $|\hype^0_i \cap M|=1$. Let $q_1, \dots, q_k$ be the vertices in~$G$ 
such that $e^0(q_i)$ is the unique edge in~$\hype^0_i \cap M$. Our aim is to show that the vertices $q_1,\dots,q_k$ form a clique in~$G$.

\begin{claim} 
\label{clm:whard-induction}For each $j \in [\bar{m}]$ and $i \in [k]$ it holds that 
(i) $|\hype^j_i \cap M|=1$, and 
(ii) $\hype^j_i \cap M=\{e^j(q_i)\}$.  %
\end{claim}
\begin{claimproof}%
By the paragraph preceding the claim, both statements hold for $j=0$.
We prove our claim by induction, so assume that (i) and (ii) hold for $j-1$. 

Fix an index $i \in [k]$.
Let $\underrightarrow{V}_i=(v_1,\dots,v_{|V_i|})$, and
assume that $q_i$ is the $p$-th vertex in this order, so by (ii) we have $\hype^{j-1}_i \cap M=\{e^{j-1}(v_p)\}$.
Observe that none of the edges $\hat{e}^j(v_1), \dots,\hat{e}^j(v_{p-1})$ is dominated at the vertex~$r_{(i,e_j)}^1$.
Since $r_{(i,e_j)}^1$ is incident to $k$ edges in~$M \cap \hype^{j-1}$ by (i) and the capacity of $r_{(i,e_j)}^1$ is~$k$,
we know that $\hat{\hype}_i^j \cap M = \emptyset$. 
Therefore, all of the edges $\hat{e}^j(v_1), \dots,\hat{e}^j(v_{p-1})$ must be dominated at $r_{(i,e_j)}^2$. 
By the preferences of~$r_{(i,e_j)}^2$ (and using $\hat{\hype}^j \cap M = \emptyset$), 
it follows that $M$ must contain an edge in $\hype_i^j$ that precedes $\hat{e}^j(v_{p-1})$ in~$\underleftarrow{\hat{\hype}}_i^{j} \merge \underleftarrow{\hype}_i^j$. 
By definition, this means that $M$ contains $e^j(v_{p'})$ for some $p' > p-1$. 
By $\capac(r^2_{(i,e_j)})=k+1$, this also yields that $r^2_{(i,e_j)}$ is saturated by $M$, 
so $|M \cap \hype_i^j|=1$. 

Consider now the vertex $r^3_{(i,e_j)}$ and its capacity of~$k+1$. 
We already know that $|M \cap (\hype^{j-1} \cup  \hype_i^j)|= k+1$, and since every edge in $\hype^{j-1} \cup  \hype_i^j$ is incident to~$r^3_{(i,e_j)}$, 
we know that $M \cap \hh{\hype}_i^j=\emptyset$. Thus, every edge $\hh{e}^j(v_h) \in \hh{\hype}_i^j$ must be dominated either at $r^3_{(i,e_j)}$ or at $r^4_{(i,e_j)}$.
In the former case, $e^{j-1}(v_p)$ precedes $\hh{e}^j(v_h)$ in~$\underleftarrow{\hh{\hype}}_i^j \merge \underleftarrow{\hype}_i^{j-1}$, meaning $h<p$.
In the latter case, $e^j(v_{p'})$ precedes $\hh{e}^j(v_h)$ in~$\underrightarrow{\hype}_i^j \merge \underrightarrow{\hh{\hype}}_i^j$, meaning $p' \leq h$. 
We can conclude that $h<p$ or $p' \leq h$ holds for \emph{every} $h \in [|V_i|]$, in particular, for $h=p$ this yields $p' \leq p$. 
Taking into account that we proved $p' \geq p$, we get $p'=p$. Hence, we have proved that $\hype^j_i \cap M=\{e^j(v_p)\} = \{e^j(q_i)\}$,
finishing our proof of the claim. 
\end{claimproof}

Using Claim~\ref{clm:whard-induction}, we can now show that the vertices $q_1,\dots,q_k$ form a clique in~$G$. 
Suppose for the sake of contradiction that $q_{i_1}$ and $q_{i_2}$ are not adjacent in~$G$ for some $i_1<i_2$. 
Then $\bar{E}$ contains the edge $\{q_{i_1},q_{i_2}\}$, let $e_j=\{q_{i_1},q_{i_2}\}$. Consider the vertices $a_{e_j}$ and $a'_{e_j}$. 
By our claim, we know that they are incident to exactly $k$ edges in $M \cap \hype^j$, and since they have capacity~$k$, 
it follows that the edge $f_{e_j}$ (incident to both $a_{e_j}$ and $a'_{e_j}$) cannot be contained in~$M$. 
Hence, it must be dominated either at $a_{e_j}$ or at $a'_{e_j}$. 
The former happens exactly if $e^j(q_{i_1}) \notin (M \cap \hype^j)$, while the latter happens exactly if $e^j(q_{i_2}) \notin (M \cap \hype^j)$. 
However, by definition, neither of these holds, a contradiction.

\smallskip
\noindent
{\bf Direction ``$\Leftarrow$'':}
Suppose now that there is a clique of size~$k$ in~$G$, containing a vertex $q_i$ from each partition~$V_i$.
We define the $\capac$-mathing~$M=\{e_w\} \cup \{ e^j(q_i): i \in [k], j \in \{0,1, \dots,\bar{m}\} \}$. 
It is easy to check that $M$ saturates all vertices of~$\HH$. 

It suffices to show that $M$ is stable. To this end, we argue that no edge of~$\HH$ blocks~$M$:
\begin{itemize}
\item Edges in~$\{d_1,\dots,d_k\}$:
It is clear that any edge~$d_i$ is dominated by~$M$ at~$s_i$, so cannot block~$M$. \\[-10pt]
\item Edges in $\hat{\hype}^j$ for some $j \in [\bar{m}]$:
Consider an edge~$\hat{e}^j(v) \in \hat{\hype}^j$ with $v \in V_i$. 
On the one hand, if $q_i$ weakly precedes $v$  in $\underrightarrow{V}_i$
(allowing $q_i=v$), then $\hat{e}^j(v)$  is dominated by~$M$ at $r_{(i,e_j)}^1$:
indeed, $r_{(i,e_j)}^1$ prefers each of the $k-1$ vertices in $M \cap (\hype_{<i}^{j-1} \cup \hype_{>i}^{j-1})$ to~$\hat{e}^j(v)$, 
and since $e^{j-1}(q_i)$ precedes~$\hat{e}^j(v)$ in~$\underrightarrow{\hype}_i^{j-1} \merge \underrightarrow{\hat{\hype}}_i^j$, 
we get that $r_{(i,e_j)}^1$ prefers $e^{j-1}(q_i)$ to $\hat{e}^j(v)$ as well. 
On the other hand, if $q_i$ follows~$v$ in~$\underrightarrow{V}_i$, then $\hat{e}^j(v)$  is dominated by~$M$ at~$r_{(i,e_j)}^2$:
again, $r_{(i,e_j)}^2$ prefers each of the $k$ vertices in $M \cap \hype^{j-1}$ to~$\hat{e}^j(v)$, 
and since $e^{j-1}(q_i)$ precedes~$\hat{e}^j(v)$ in~$\underleftarrow{\hat{\hype}}_i^j \merge \underleftarrow{\hype}_i^j$,
we get that $r_{(i,e_j)}^2$ prefers~$e^{j-1}(q_i)$ to~$\hat{e}^j(v)$ as well. 
\item Edges in $\hh{\hype}^j$ for some $j \in [\bar{m}]$: 
Consider an edge $\hh{e}^j(v) \in \hat{\hype}^j$ and $v \in V_i$. We argue similarly as in the previous case.
If $q_i$ weakly precedes $v$ in~$\underrightarrow{V}_i$, then $\hh{e}^j(v)$ is dominated by~$M$ at~$r_{(i,e_j)}^4$, 
otherwise it is dominated by~$M$ at~$r_{(i,e_j)}^3$.
\item Edges in~$\{f_{e_j} :e_j \in \bar{E}\}$: Let $e_j$ run between $V_{i_1}$ and $V_{i_2}$ in the complement of~$G$, with $i_1<i_2$. 
Since $e_j$ is not an edge of~$G$, at least one of its endpoints does not belong to the clique on~$\{q_1,\dots,q_k\}$. 
If this endpoint is a vertex in $V_{i_1} \setminus \{q_{i_1}\}$, then $f_{e_j}$ is dominated at~$a_{e_j}$ by the edges $\{e^j(q_i):i \in [k]\}$.
By contrast, if $e_j$ has an endpoint in~$V_{i_2} \setminus \{q_{i_2}\}$, then $f_{e_j}$ is dominated at $a'_{e_j}$. 
\item Edges in $\hype^j$ for some $j \in \{0,1,\dots,\bar{m}\}$: 
Suppose that $v \in V_i$, and consider the edge $e^j(v)$. 
On the one hand, if $v$ precedes $q_i$ in~$\underrightarrow{V}_i$, then $e^j(v)$ is dominated at $r_{(i,e_j)}^2$ (in case $j \in [\bar{m}]$)
and also at~$r_{(i,e_{j+1})}^3$ (in case $0 \leq j <\bar{m}$). 
On the other hand, if $v$ follows $q_i$ in~$\underrightarrow{V}_i$, then $e^j(v)$ is dominated at~$r_{(i,e_j)}^1$ (in case $0 \leq j <\bar{m}$) 
and also at~$r_{(i,e_{j+1})}^4$ (in case $j \in [\bar{m}]$). \qedhere
\end{itemize}
\end{proof}

\section{Omitted proofs from Section~\ref{sec:subtree}}

\subsubsection*{Proof of Theorem~\ref{thm:fa}} 
\label{sec:app-proof-thmfa}

\thmfa*
\begin{proof}
We give a reduction from \shbm, which is $\NP$-complete even if each hyperedge has size~3~\cite{NgHirschberg91}. Let $I=(\HH,\capac,(\succ_v)_{v \in V})$ be such an instance of \shbm\ with hypergraph $\HH=(V,\EE)$. We create a new instance $I'=(\HH',\capac',(\succ'_v)_{v \in V})$ of \subtrshbm\ satisfying the conditions in the theorem. Let the vertices of $\mathcal{H'}$ be the vertices of $\HH$ with one additional vertex~$x$, and let us extend each hyperedge in~$\EE$ by adding~$x$. That is, $\HH'=(V \cup \{x\},\EE')$ where $\EE'=\{e \cup \{x\}:e \in \EE\}$.
Notice that $\HH'$ is indeed a subtree hypergraph whose underlying tree is a star on~$V \cup \{x\}$ with $x$ as its center.
For each vertex $v \in V$, we retain both its capacity and its preferences, i.e., $\capac'(v)=\capac(v)$ and $\succ'_v=\succ_v$. For the vertex~$x$, we set $\capac'(x)=|\EE|+1$ and fix its preferences arbitrarily. Note that 
$x$ can never be saturated in any $\capac$-matching $M \subseteq \EE'$, so $x$ cannot dominate any edge for any $\capac$-matching. 

It is now straightforward to check that the stable $\capac$-matchings in~$I$ correspond bijectively to the stable $\capac$-matchings in~$I'$. 
More precisely, for a stable $\capac$-matching~$M'$ in~$I'$, the projection of~$M'$ to~$V$ yields a stable $\capac$-matching in~$I$. Similarly, 
given a stable $\capac$-matching~$M$ in~$I$, the $\capac$-matching obtained by adding $x$ to each hyperedge in~$M$ is a stable $\capac$-matching in~$I'$. 
\end{proof}

\end{appendices}
\end{document}